\documentclass[11pt]{article}

\usepackage[T1]{fontenc}
\usepackage{lmodern}
\usepackage[margin=1in]{geometry}
\usepackage{amsmath,amssymb,amsfonts,amsthm}
\usepackage{enumerate}
\usepackage[natbibapa,bibnewpage]{apacite}
\usepackage{booktabs}
\usepackage{graphicx}
\usepackage{caption}
\usepackage{subcaption}
\usepackage{xcolor}
\usepackage{tikz}
\usepackage{pgfplots}
\usetikzlibrary{calc,patterns,positioning,arrows.meta,decorations.pathreplacing,decorations.markings}
\usepgfplotslibrary{fillbetween}
\pgfplotsset{compat=1.18}
\usepackage{setspace}
\usepackage{xurl}
\usepackage[colorlinks=true,citecolor=blue,linkcolor=blue,urlcolor=blue]{hyperref}
\AtBeginDocument{%
  \renewcommand{\doi}[1]{\url{https://doi.org/#1}}%
}
\hypersetup{
  pdftitle={Equilibrium in closed constant-function market maker economies},
  pdfauthor={Muqiao Huang, Ruodu Wang, and Yiyun Wang},
  pdfkeywords={decentralized finance, automated market makers, constant-function market makers, unilateral no-trade equilibrium, homothetic preferences, reachability}
}

\theoremstyle{plain}
\newtheorem{theorem}{Theorem}
\newtheorem{proposition}{Proposition}
\newtheorem{lemma}{Lemma}
\newtheorem{corollary}{Corollary}

\theoremstyle{definition}
\newtheorem{definition}{Definition}
\newtheorem{example}{Example}

\theoremstyle{remark}
\newtheorem{remark}{Remark}

\newcommand{\dd}{\,\mathrm{d}}
\newcommand{\R}{\mathbb{R}}
\newcommand{\N}{\mathbb{N}}
\newcommand{\proofpart}[1]{\par\smallskip\noindent\emph{#1.}}

\begin{document}

\title{Equilibrium in closed constant-function market maker economies}
\date{}
\author{Muqiao Huang\thanks
 {Department of Statistics and Actuarial Science,
 University of Waterloo, Canada.
 \href{mailto:m5huang@uwaterloo.ca}{m5huang@uwaterloo.ca}.} \and Ruodu Wang\thanks
 {Department of Statistics and Actuarial Science,
 University of Waterloo, Canada.
 \href{mailto:wang@uwaterloo.ca}{wang@uwaterloo.ca}.}
 \and Yiyun Wang\thanks
 {Department of Statistics and Actuarial Science,
 University of Waterloo, Canada.
 \href{mailto:y4658wan@uwaterloo.ca}{y4658wan@uwaterloo.ca}.}}

\maketitle

\begin{abstract}
We study equilibria in a closed, fee-free constant-function market maker (CFMM) economy with two assets and two agents. An interior state is a unilateral no-trade equilibrium exactly when the CFMM marginal price equals both agents' marginal rates of substitution. For an interior initial state, individually rational unilateral equilibria are Pareto optimal relative to the fixed CFMM invariant. A weak representative agent is obtained at each fixed equilibrium by weighted sup-convolution, whereas, under some regularity assumptions, a state-independent strong representative agent exists exactly when agents share the same homothetic preference. Every interior feasible state is reachable through finitely many valid trades, but an individually rational interior target need not be improvement reachable. Alternating utility-maximizing trades converge to a Pareto optimal unilateral equilibrium.

\bigskip

\noindent\textbf{Keywords:} Decentralized finance, constant-function market makers, unilateral no-trade equilibrium, homothetic preferences, reachability.

\end{abstract}

\section{Introduction}\label{sec:Intro}

Automated market makers (AMMs) have become a central trading mechanism in decentralized finance, but their origin predates blockchain markets. \citet{H03} introduced an automated liquidity supplier for thin prediction markets, and Hanson's logarithmic market scoring rule provides a canonical example in which prices are generated mechanically from a potential function rather than by matching contemporaneous buyers and sellers; see also \citet{R2007}. Despite the important differences between prediction markets and decentralized asset markets, the two settings share a close mathematical structure; see the more recent studies  of \cite{FPW23} and \cite{SKM23}.

The dominant blockchain implementation is the constant-function market maker (CFMM), with Uniswap as a leading example \citeauthor{AZR2020}~(\citeyear{AZR2020}, \citeyear{AZSKR2021}). A CFMM carries an inventory vector and accepts trades that preserve a prescribed invariant. Because its inventory determines the pool's current trading conditions, it naturally serves as an economic state variable. Each trade changes not only the allocation of assets but also the price and the terms available to subsequent traders. This differs fundamentally from a Walrasian market, in which an individual agent takes a fixed price vector as given, and from a limit order book, in which prices emerge from matching orders and liquidity suppliers actively manage quotes. Liquidity providers of a CFMM are commonly compensated through pro rata trading fees, whereas liquidity suppliers in order-book markets seek spread and, in some venues, rebate revenues while bearing inventory and adverse-selection risk; see \citet{BCCGG2023,LP2025}.

Our paper's main aim is to analyze a market with a CFMM from an economic equilibrium point of view, and in particular to focus on its difference from the classic Walrasian market equilibrium. As a starting point for such a theory, we specialize to a closed, fee-free benchmark with one CFMM, two assets, and two agents. There is no external trading venue, and the agents do not provide liquidity. The closed-market assumption is a modeling restriction, and of course decentralized asset markets are typically not isolated in practice. Its purpose is to identify the internal equilibrium force generated by the liquidity curve, and to offer insights on this idealized setting  that are generalizable to more realistic markets.  
Our study is not meant to be  comprehensive, and we do not consider  important  factors such as fees, outside arbitrage, and strategic liquidity provision, which are practical but would obscure the basic mechanism that we aim to study.

In addition to proposing  the new   equilibrium framework of the closed CFMM economy, the paper's contribution is threefold. First, we characterize the unilateral equilibrium and connect it to Pareto optimality by establishing  analogues of the  Welfare Theorems in our framework. 
Second, our results offer insights on the welfare and aggregation
consequences of replacing a fixed budget line by a nonlinear liquidity
curve through a study of representative agents. Third, our results separate mechanical, individually rational, and
best-response reachability and show how sequential best responses select a
unilateral equilibrium. These contributions can be explained through a few questions that naturally arise.

Our first question is what equilibrium means when all trades take place through a CFMM and the terms of trade depend on its current inventory. At a fixed market state, the status quo action of each agent is the zero trade, and a unilateral deviation is any valid trade against the current CFMM inventory while the other agent does not move. We call a market state a unilateral no-trade equilibrium (Definition~\ref{def:4}), or simply a unilateral equilibrium, when there is no profitable unilateral deviation. Under Assumptions~\eqref{as:1}--\eqref{as:4} in Section~\ref{sec:uf}, an interior state is a unilateral equilibrium if and only if the CFMM marginal price equals both agents' marginal rates of substitution (Theorem~\ref{thm:1}). This condition resembles the familiar equalization of marginal rates at an interior Walrasian equilibrium. Its economic interpretation is different, however, because it describes tangency to the nonlinear feasible frontier induced by the CFMM rather than optimization over a common linear budget set. For an interior initial state, every individually rational unilateral equilibrium is Pareto optimal relative to the fixed CFMM invariant (Theorem~\ref{thm:unilateral-implies-Pareto}), giving an analogue of the First Welfare Theorem.
Although we specialize in the CFMM market setting, these results are rather standard in economic analysis; see \cite{GWWY2024} in the setting of prediction market, which have similar results.

We next ask when the behavior of several agents can be summarized by a representative agent. This question is related to the classical aggregation problem studied by \citeauthor{Gorman1953}~(\citeyear{Gorman1953}, \citeyear{Gorman1961}). After a particular equilibrium has been fixed, a weak representative agent is obtained from a weighted sup-convolution of the agents' utilities. The weights are chosen to align the agents' utility gradients at the equilibrium allocation. The more demanding question is whether the same representative utility can work across different initial states and equilibria. Within the class of utility functions considered here, such a state-independent representation exists only under a restrictive geometric condition. The two-parameter family of aggregate frontiers generated by infimal convolution must collapse to the one-parameter level-curve family of a single utility. We show that this occurs exactly when the agents share a common homothetic preference (Theorem~\ref{thm:strong-rep}).

Finally, we study reachability and trading dynamics. From an interior initial state, every interior feasible state can be implemented by a finite sequence of valid CFMM trades (Theorem~\ref{thm:interior-feasible-reachable}). Alternating utility-maximizing trades converge to a Pareto optimal unilateral equilibrium (Theorem~\ref{thm:max-reach-limit}). Trading order matters because each trade changes the state faced by the next trader. In the first round, a first-mover disadvantage arises when the agents' initial indifference prices lie on opposite sides of the initial CFMM price, while a first-mover advantage arises when they lie on the same side and the first-trade no-crossing condition holds. In the latter case, the advantage persists at the limiting equilibria, and each agent's first trade when moving first already yields more utility than her limiting holding when moving second (Theorem~\ref{thm:first-mover-advantage-limit}).

These results identify several genuinely non-Walrasian phenomena. The aggregate holdings of the agents need not remain fixed because the CFMM is an inventory-bearing participant, even though total system inventory is conserved. Equilibrium allocations may therefore differ from the corresponding Walrasian allocation even when marginal prices agree. In a numerical example, both agents attain higher utility at a unilateral equilibrium than at the corresponding Walrasian equilibrium; see Example~\ref{ex:liquidity-welfare}. This does not mean that the CFMM creates resources. The two mechanisms compare different feasible sets because the pool supplies one asset and absorbs the other. The model also generates path dependence, reachability restrictions, and transaction-order effects solely through endogenous movement along the liquidity curve.
The take-home message of our study is 
\begin{center}
    \emph{CFMMs offer a tractable framework for the economics of exchange with nonlinear budget curves, and its equilibria, welfare, preference aggregation, and trading-order effects can be quantitatively characterized.}
\end{center}

The remainder of the paper is organized as follows. At the end of the Introduction we discuss some related literature. Section~\ref{sec:pre} reviews the Walrasian benchmark, introduces CFMMs, and compares budget-line trading with liquidity-curve trading. Section~\ref{sec:uf} develops the utility and level-curve results used throughout. Section~\ref{sec:eq0} studies static unilateral equilibria and welfare. Section~\ref{sec:cba} treats aggregation and representative agents. Section~\ref{sec:tr} studies reachability, alternating best responses, and first-mover effects. Section~\ref{sec:con} concludes. All proofs of the main results are collected in Sections \ref{app:technical-proofs}--\ref{ec:pf6}.

\textbf{Related literature.} 
The closest connections concern CFMM price impact and trader
optimization. \citet{BF24} analyze how pool size, pooling, and fees
affect price impact, while \citet{CDM2025} study optimal execution
and statistical arbitrage in constant-product markets under stochastic
dynamics. We consider two utility-maximizing traders who interact in the CFMM market (this is our key feature), and study their behaviour through
a fixed liquidity curve, determining endogenous equilibrium prices
and allocations, and establish the resulting welfare properties. \citet{GWWY2024} study repeated utility-maximizing trades in
utility-based prediction markets, establishing convergence to
Pareto-efficient allocations and examining sequence-dependent prices.
They obtain similar results to our 
Theorems \ref{thm:1}--\ref{thm:unilateral-implies-Pareto}
and Proposition \ref{prop:weak-rep} (these results are standard in economic analysis)
in their market setting. 
\citet{SKM23} characterize CFMMs axiomatically, and \citet{FPW23}
connect them to prediction markets; \citet{ADPWY2011} unify
prediction-market mechanisms through convex optimization.
\citet{ACDEK2023} and \citet{BFP25} study the composition
of liquidity sources; our aggregation characterizes
when trader preferences admit a state-independent
representative utility.
\citet{Park2023} analyzes sandwich attacks and trading regret, while
\citet{XMP2023} study verifiable sequencing; we derive first-mover
utility bounds from pool inventory changes alone.
On the liquidity-provider side, \citet{CJ2021} and \citet{LP2025}
study liquidity provision and pool size, while
\citet{HRS2025} analyze concentrated liquidity and \citet{HRS2026}
examine how fees affect liquidity supply and trading volume.
\citet{AD2021} and \citet{MMRZ2022} analyze impermanent loss and
loss-versus-rebalancing, respectively. Our welfare results concern
traders facing a fixed CFMM invariant.
For background, see \citet{AKCNC2019} on Uniswap markets,
\citet{CK2019} on blockchain settlement, and the
review by \citet{BCCGG2023}.

\section{Preliminaries}\label{sec:pre}
\subsection{Classical pure exchange economy} \label{sub:ex}
Classical pure exchange economies provide a useful benchmark for studying equilibrium prices and allocations. In this framework, agents trade assets at exogenously given prices and are modeled as price takers. This abstraction is appropriate for competitive markets in which no individual agent can affect prices, and it leads to linear budget constraints and market-clearing conditions. We recall this benchmark because it highlights the main departure of the CFMM setting studied in this paper: in a CFMM market, trades move the CFMM inventory and therefore change the CFMM price.

We consider a pure exchange economy as in \citet[Section 15.B]{MWG1995}, with two consumer agents $i=1,2$ and two tradable assets, denoted by $X$ and $Y$. In the Walrasian equilibrium terminology, these would be referred to as commodities, but here we adopt a financial interpretation.
We write $\R_+=[0,\infty)$, $\R_{++}=(0,\infty)$, $\bar\R=[-\infty,\infty]$, and $\N=\{0,1,2,\ldots\}$.
An allocation $\left(\mathbf r_1, \mathbf r_2\right)\in \mathbb{R}_{+}^2 \times \mathbb{R}_{+}^2$ specifies each agent's holdings, where $\mathbf r_i= \left(x_i, y_i\right)$ represents agent $i$'s holdings in the two assets.
The restriction of $\mathbf r_i\in \R_+^2$ reflects that short positions are not allowed on these assets.
We also denote the initial endowments by $(\mathbf r_1^{\mathrm{in}}, \mathbf r_2^{\mathrm{in}}) \in \mathbb{R}_{+}^2 \times \mathbb{R}_{+}^2$.
In this subsection, unsuperscripted holdings denote the allocation under consideration; superscripts label initial or comparison allocations.
Suppose that each agent's preferences can be represented by a utility function $u_i: \mathbb{R}_{+}^2 \rightarrow \mathbb{R}$, for $i=1,2$.
Let
\[
\Delta_n=\left\{\left(x_1, \ldots, x_n\right) \in \mathbb{R}_{+}^n: \sum_{i=1}^n x_i=1\right\}
\]
denote the standard simplex in $\mathbb{R}^n$. The Walrasian equilibrium in a pure exchange economy is then defined as follows.

\begin{definition}\label{def:1}
Given the initial endowments $(\mathbf r_1^{\mathrm{in}},\mathbf r_2^{\mathrm{in}})\in \mathbb{R}_+^2 \times \mathbb{R}_+^2$, a triple of vectors
\[
(\mathbf r_1^*, \mathbf r_2^*, \mathbf q^*) \in \mathbb{R}_+^2 \times \mathbb{R}_+^2 \times \Delta_2
\]
is called a \emph{Walrasian equilibrium} if
\begin{itemize}
\item[(a)] for each agent $i=1,2$, the allocation $\mathbf r_i^*$ solves
\[
\max_{\mathbf r_i \in \mathbb{R}_+^2} u_i(\mathbf r_i)
\quad \text{subject to} \quad
\mathbf q^* \cdot \mathbf r_i \leq \mathbf q^* \cdot \mathbf r_i^{\mathrm{in}};
\]
\item[(b)] $\mathbf r_1^* + \mathbf r_2^* = \mathbf r_1^{\mathrm{in}} + \mathbf r_2^{\mathrm{in}}$.
\end{itemize}
\end{definition}

Walrasian equilibrium assumes a linear trading constraint, where assets are priced linearly according to $\mathbf q^*$. In the standard economic setting of a competitive market economy with multiple agents and homogeneous assets, linear pricing emerges because agents are price-takers: no individual agent has sufficient market power to influence prices, and transactions occur at uniform rates across the market, leading to a single price vector that clears supply and demand. Condition (a) states that, in equilibrium, each agent maximizes their utility subject to a budget constraint, while condition (b) requires that the market clears. Allocations satisfying condition (b) are called \emph{feasible}, as they ensure that the total allocation equals the aggregate initial endowments. The vector $\mathbf q^*=(q_1^*,q_2^*)$ is the equilibrium price, at which agents can exchange between assets $X$ and $Y$.

\begin{definition}
An allocation $(\mathbf r_1^\alpha,\mathbf r_2^\alpha)\in\R_+^2\times\R_+^2$
is \emph{Pareto optimal} if it is feasible, that is,
$\mathbf r_1^\alpha+\mathbf r_2^\alpha
=\mathbf r_1^{\mathrm{in}}+\mathbf r_2^{\mathrm{in}}$, and there exists
no other feasible allocation $(\mathbf r_1^\beta,\mathbf r_2^\beta)$ such
that $u_i(\mathbf r_i^\beta)\geq u_i(\mathbf r_i^\alpha)$ for both
$i=1,2$, with strict inequality for at least one $i$.
\end{definition}

Walrasian equilibria are often Pareto optimal, and results of this type are commonly referred to as the First Fundamental Theorem of Welfare Economics, which asserts that under mild assumptions, competitive equilibria achieve Pareto optimality.
For precise statements of this theorem, see \citet[Proposition 16.C.1]{MWG1995} and \citet[Box 2.9]{L2004}.

Assume that $u_1$ and $u_2$ are differentiable and strictly increasing, with nonzero partial
derivatives. For an agent with utility $u$ and holdings $(x,y)$, a marginal
increase $\epsilon>0$ in asset $X$ changes utility by approximately
$\partial_x u(x,y)\epsilon$. Up to higher-order terms, the same utility change
can be obtained by increasing asset $Y$ by
\[
\frac{\partial_x u(x,y)}{\partial_y u(x,y)}\epsilon.
\]
Accordingly, define the marginal rate of substitution, or indifference price
of $X$ in units of $Y$, by
\begin{equation}\label{eq:def-p}
p(u_i,x,y):=\frac{\partial_x u_i(x,y)}{\partial_y u_i(x,y)}.
\end{equation}
At an interior Walrasian equilibrium $(\mathbf r_1^*,\mathbf r_2^*,\mathbf q^*)$,
\begin{equation}\label{eq:mar}
p(u_i,\mathbf r_i^*)=\frac{q_1^*}{q_2^*},\qquad i\in\{1,2\}.
\end{equation}
Otherwise, agent $i$ could trade at the fixed price $q_1^*/q_2^*$ and
strictly increase utility, contradicting optimality. Because utilities are
strictly increasing, both components of the equilibrium price vector are
positive.

\subsection{Constant-function market makers}\label{sec:cfmm}

A CFMM with $n$ assets is specified by a trading function $\varphi:\R_+^n\to\R$ and a current inventory $\mathbf r_0\in\R_+^n$. When an agent proposes $\mathbf h\in\R^n$, the CFMM accepts the trade exactly when
\[
\mathbf r_0+\mathbf h\in\R_+^n,
\qquad
\varphi(\mathbf r_0+\mathbf h)=\varphi(\mathbf r_0).
\]
The inventory is then updated to $\mathbf r_0+\mathbf h$. Thus every accepted trade preserves the invariant value of $\varphi$.

Although this definition allows any number of assets, the rest of the paper uses two assets, denoted by $X$ and $Y$, to match the pure-exchange benchmark. At invariant level $c\in\R$, the admissible CFMM inventories form the liquidity curve
\[
s_{\varphi,c}:=\{\mathbf r\in\R_+^2:\varphi(\mathbf r)=c\}.
\]
When this set is represented as a graph, we reserve $\ell_{\varphi,c}$ for its graphing function.

Within CFMMs, successive modifications of the trading function have given rise to well-known AMMs used in practice. Uniswap popularized the two-asset constant-product design, StableSwap introduced an invariant tailored to assets expected to trade near parity, and Balancer developed weighted constant-mean pools that can contain more than two assets; see \cite{AZR2020,Egorov2019,MartinelliMushegian2019}. We next introduce several trading functions of those well-known CFMMs.
In what follows, $\mathbf r=(x,y)\in \R_+^2$.

\begin{itemize}
 \item \textbf{Uniswap V2} takes the form
 \[
\varphi(\mathbf r) = \log(x)+\log(y),
\]
where slippage and divergence loss remain major concerns for traders and liquidity providers. This logarithmic representation has the same level curves as the constant-product invariant in \cite{AZR2020}.

\item \textbf{Uniswap V3} introduces concentrated liquidity through the use of virtual inventory, and its trading function takes the form
\[
\varphi(\mathbf r) = \log(x + a)+\log(y + b), \quad a,\,b > 0.
\]
Here, $a$ and $b$ represent virtual inventory, allowing liquidity to be effectively concentrated within a specified price range. This formula should be interpreted as a local shifted-reserve representation within an active liquidity range, as in the virtual-reserve description of concentrated liquidity in \citet{AZSKR2021}. Once a boundary of the active range is reached, the feasible liquidity curve is truncated, and the interior first-order price condition has to be replaced by a one-sided condition; see the discussion in Section~\ref{sec:eq0}.
\item
\textbf{StableSwap} is designed for assets expected to trade near parity. A convenient two-asset stylized StableSwap-type trading function is
\[
\varphi(\mathbf r)=\log\bigl(C(x+y)+xy\bigr),\qquad C>0.
\]
At a fixed level, this is the invariant $C(x+y)+xy=\mathrm{constant}$. It captures the interpolation between near-constant-sum behavior and the reserve protection supplied by a product term, but it should not be identified with the full multi-asset Curve invariant. See \citet{Egorov2019} for the original StableSwap construction and \citet[Section EC.2.5.1]{BF24} for this two-asset representation.
\item \textbf{Balancer} takes the form
\[
\varphi(\mathbf r) = \theta\log(x)+(1-\theta)\log(y),
\]
with $\theta\in(0,1)$. It is a weighted version of Uniswap V2. This is the two-asset case of the weighted-pool construction in \citet{MartinelliMushegian2019}.
\end{itemize}

The form of a trading function can vary widely, but important market properties such as the behavior of the pricing oracle or the presence of divergence loss can still be derived in general from its mathematical structure. For example, \citet{BF24} proposed axioms on the underlying utility function and thereby obtained results on the size of swaps and the induced pricing oracle. \citet{SKM23} focused particularly on separability and on different invariance properties under scaling to identify specific subclasses of CFMMs. Finally, \citet{FPW23} introduced an axiomatic framework for market makers and proved that CFMMs equipped with concave trading functions satisfy the desired properties of a good market maker.

If $\varphi$ is differentiable with nonzero partial derivatives, define the
CFMM price of $X$ in units of $Y$ by
\[
p(\varphi,x,y):=\frac{\partial_x \varphi(x,y)}{\partial_y \varphi(x,y)}.
\]
\citet{SKM23} and \citet{FPW23} instead use exchange rates relative to the
``grand bundle'' of all assets. In the two-asset case, their normalized
representation is
\[
\left( \frac{ p(\varphi, x,y) }{1+ p(\varphi, x,y)} \,,\, \frac{1}{1+ p(\varphi, x,y)} \right).
\]

\subsection{Price-taking and liquidity-curve trading}

We first clarify the sense in which the present framework departs from the classical pure exchange economy. In a Walrasian equilibrium, each agent takes a price vector $\mathbf q$ as given and maximizes utility over the budget set. The boundary of this set is a straight budget line, and the price vector is not affected by the agent's individual trade.

In a CFMM market, the trading constraint is generated instead by the liquidity curve of the market maker. If the current CFMM inventory is $\mathbf r_0=(x_0,y_0)$, then a trade $\mathbf h=(h_x,h_y)$ is accepted by the CFMM if and only if
$\varphi(\mathbf r_0+\mathbf h)=\varphi(\mathbf r_0)$.
Hence a trade moves the CFMM inventory along the liquidity curve $s_{\varphi, \varphi(\mathbf r_0)}$.
At each point of this curve, the CFMM price is given by $p(\varphi,x,y)$.
The terms of trade therefore vary with the size and direction of the trade.

Given the current CFMM inventory \(\mathbf r_0\) and agent holding
\(\mathbf r_i\), the agent's feasible frontier is
\[
\left\{
\mathbf r_i-\mathbf h:
\mathbf r_0+\mathbf h\in\mathbb R_+^2,\
\mathbf r_i-\mathbf h\in\mathbb R_+^2,\
\varphi(\mathbf r_0+\mathbf h)=\varphi(\mathbf r_0)
\right\}.
\]
Unlike the Walrasian budget line, this feasible frontier is generally a
 nonlinear curve.

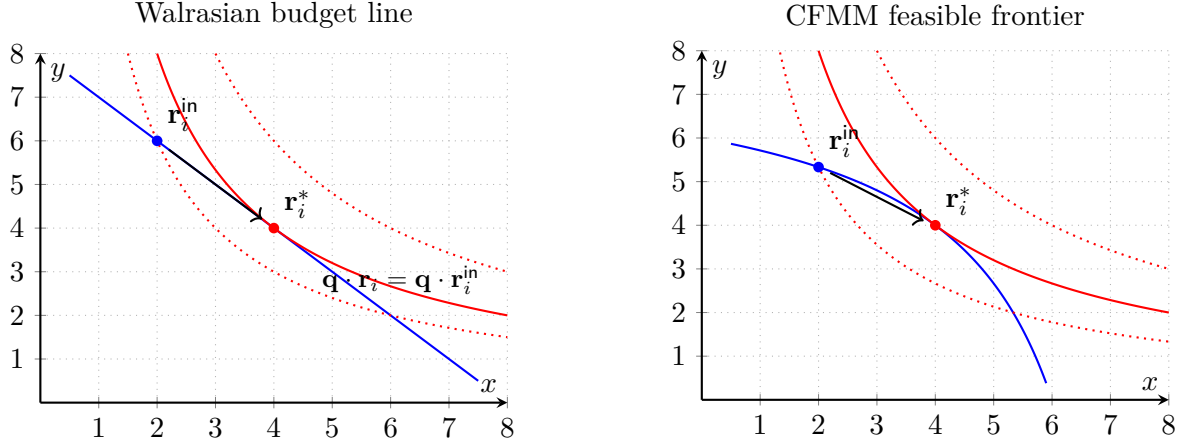
\begin{figure}[htbp]
\centering
\begin{minipage}{0.47\textwidth}
\centering
\begin{tikzpicture}
\begin{axis}[
 width=\linewidth,
 height=6.2cm,
 xlabel={$x$},
 ylabel={$y$},
 xmin=0, xmax=8,
 ymin=0, ymax=8,
 axis lines=middle,
 axis line style={thick, black},
 grid=major,
 grid style={dotted, gray!60},
 xtick={0,1,2,3,4,5,6,7,8},
 ytick={0,1,2,3,4,5,6,7,8},
 samples=300,
 smooth,
 no markers,
 restrict y to domain=0:8,
 title={Walrasian budget line}
]
\addplot[blue, thick, domain=0.5:7.5] {8-x};
\addplot[red, dotted, thick, domain=1.5:8] {12/x};
\addplot[red, thick, domain=2:8] {16/x};
\addplot[red, dotted, thick, domain=3:8] {24/x};
\fill[blue] (axis cs:2,6) circle[radius=2pt];

\fill[red] (axis cs:4,4) circle[radius=2pt];

\node[anchor=south west]
 at (axis cs:2,6)
 {$\mathbf r_i^{\mathrm{in}}$};

\node[anchor=south west]
 at (axis cs:4,4)
 {$\mathbf r_i^*$};

\node[
 anchor=south east,
 font=\small
]
 at (axis cs:7.7,2.25)
 {$\mathbf q\cdot\mathbf r_i
 =\mathbf q\cdot\mathbf r_i^{\mathrm{in}}$};
\draw[->, thick]
 (axis cs:2.2,5.8)
 --
 (axis cs:3.8,4.2);

\end{axis}
\end{tikzpicture}
\end{minipage}
\hfill
\begin{minipage}{0.47\textwidth}
\centering
\begin{tikzpicture}
\begin{axis}[
 width=\linewidth,
 height=6.2cm,
 xlabel={$x$},
 ylabel={$y$},
 xmin=0, xmax=8,
 ymin=0, ymax=8,
 axis lines=middle,
 axis line style={thick, black},
 grid=major,
 grid style={dotted, gray!60},
 xtick={0,1,2,3,4,5,6,7,8},
 ytick={0,1,2,3,4,5,6,7,8},
 samples=300,
 smooth,
 no markers,
 restrict y to domain=0:8,
 title={CFMM feasible frontier}
]
\addplot[
 blue,
 thick,
 domain=0.5:5.9
] {8 - 16/(8-x)};
\addplot[
 red,
 dotted,
 thick,
 domain=1.34:8
] {32/(3*x)};

\addplot[
 red,
 thick,
 domain=2:8
] {16/x};

\addplot[
 red,
 dotted,
 thick,
 domain=3:8
] {24/x};
\fill[blue] (axis cs:2,5.333333333) circle[radius=2pt];

\fill[red] (axis cs:4,4) circle[radius=2pt];

\node[anchor=south west]
 at (axis cs:2,5.333333333)
 {$\mathbf r_i^{\mathrm{in}}$};

\node[anchor=south west]
 at (axis cs:4,4)
 {$\mathbf r_i^*$};
\draw[->, thick]
 (axis cs:2.2,5.2)
 --
 (axis cs:3.8,4.1);

\end{axis}
\end{tikzpicture}
\end{minipage}
\caption{
Comparison between the trading constraints faced by an agent in the
Walrasian and CFMM models. In both panels, the blue curve is the relevant
frontier and the red curves are indifference curves. In the Walrasian
model, the agent trades along a linear budget line at a fixed price. In the
CFMM model, the pool's liquidity curve induces a nonlinear, concave feasible
frontier in the agent's holdings, so the terms of trade vary with the size of
the trade.
}
\label{fig:budget-vs-liquidity}
\end{figure}

For small trades, the CFMM constraint admits a local linear approximation. If $\mathbf h$ is small, then \(\varphi(\mathbf r_0+\mathbf h)=\varphi(\mathbf r_0)\) implies \(\nabla \varphi(\mathbf r_0)\cdot \mathbf h \approx 0\), which is the tangent-line approximation to the liquidity curve at $\mathbf r_0$. Thus the Walrasian budget line can be viewed as a local, zero-price-impact approximation of the CFMM trading constraint. Our framework keeps the curvature of the liquidity curve, and therefore captures finite-liquidity price impact.

This analogy is only local. In the classical pure exchange economy, agents trade directly with each other, so their aggregate holdings remain fixed. In the CFMM economy, agents trade with the pool. Therefore the total inventory of the entire system is conserved, but the aggregate holdings of the agents alone need not be conserved. This is why the equilibrium analysis below will recover a familiar price-equality condition while allowing different allocations, welfare properties, and trading sequences.

\section{Utility functions}\label{sec:uf}
In this section, we provide some technical tools on utility functions that will be useful in our later analysis.

Let $D\subseteq\R_+^2$ be a common domain, and let
$u,v:D\to\bar\R$. The utility $u$ represents a preference by ranking
holdings according to their utility values. We call $u$ and $v$
\emph{equivalent} if, for all $\mathbf r^\alpha,\mathbf r^\beta\in D$,
\[
u(\mathbf r^\alpha)\geq u(\mathbf r^\beta)
\quad\Longleftrightarrow\quad
v(\mathbf r^\alpha)\geq v(\mathbf r^\beta).
\]
Equivalently, there is a strictly increasing bijection
$F:v(D)\to u(D)$ such that $u=F\circ v$. A CFMM trading function can be
viewed as a utility representation of the market maker's inventory
preferences.

For the formal results, let $u:\R_+^2\to\R_+$. We use the following assumptions.
\begin{enumerate}[(I)]
\item \label{as:1}
$u$ is continuous with $u(x,0) = u(0,y)= 0$ for any $x,y\in \R_+$, and strictly increasing on $\R_{++}^2$.
\item \label{as:2}$\lim_{x'\rightarrow \infty} u(x',y)=\infty$, $\lim_{y'\rightarrow \infty} u(x,y')=\infty$ for every $x,y>0$.
\item \label{as:3} $u$ is continuously differentiable and strictly concave on $\R_{++}^2$.
\item \label{as:4} for any $x_1, x_2, y_1, y_2 \in \R_{++}$ with $x_1\leq x_2, y_1\geq y_2$ with at least one inequality strict, we have $p(u,x_1, y_1)> p(u, x_2, y_2) $ whenever the prices are defined.

\end{enumerate}

Assumptions~\eqref{as:1} and~\eqref{as:3}
imply strictly positive partial derivatives
on $\R_{++}^2$. Indeed, concavity and strict
monotonicity give
$\partial_x u(x,y)\geq
[u(x+h,y)-u(x,y)]/h>0$ for $h>0$,
and similarly for $\partial_y u$.
Thus the interior prices are well defined.

The normalization to nonnegative utility values is ordinal; the boundary and curvature conditions below are substantive. Strict monotonicity says that more of either asset is preferred, while strict concavity captures increasing marginal cost along a swap. Differentiability is used to define prices. Assumption~\eqref{as:4} is a global monotonicity condition on willingness to pay. Moving southeast in the holdings plane---a larger holding of $X$ and a smaller holding of $Y$---reduces the amount of $Y$ the agent is marginally willing to exchange for one unit of $X$. Together with curvature of the CFMM liquidity curve, this condition makes the price gap in a one-agent problem strictly monotone and the best response single-valued. Assumptions~\eqref{as:1} and~\eqref{as:2} imply that positive level curves have the coordinate axes as asymptotes.

The distinction between a utility representation and its level curves will matter below. Strict concavity and differentiability are cardinal properties of the chosen representation and need not survive an increasing transform; the ordering and geometry of level sets are ordinal.

For a real-valued function $f$, write
\[
f^g:=\{(x,f(x)):x\in\operatorname{Domain}(f)\}
\]
for its graph and
\[
f^e:=\{(x,y):x\in\operatorname{Domain}(f),\ y\geq f(x)\}
\]
for its epigraph. For sets $A$ and $B$, addition is in the Minkowski sense:
$A+B:=\{a+b:a\in A,\ b\in B\}$. For $k\in\R$, write
$kA:=\{ka:a\in A\}$ whenever the scalar products are defined.

For a utility function $u:D\to\bar\R$, the level set at an attained
level $c\in u(D)$ is
\[
s_{u,c}:=u^{-1}(c)=\{(x,y)\in D:u(x,y)=c\}.
\]
We may omit $u$ or $c$ if the context is clear.
\begin{proposition}\label{prop:level-curve}
 Let $u: \R_+^2 \to \R_+$ be a utility function satisfying Assumptions \eqref{as:1} and \eqref{as:2}. Then for each $c> 0$, the level set $s_{u,c}$ is of the form $\ell^g$, the graph of a function $\ell: \R_{++} \to \R_{++}$ such that
 \begin{enumerate}[(a)]

 \item \label{curve:cts-dec} the function $\ell$ is continuous and strictly decreasing;

 \item \label{curve:asymptote}the function satisfies $\lim_{x\to 0^+} \ell(x) = \infty $ and $\lim_{x\to \infty} \ell(x) = 0$;
 \item \label{curve:ray-intercept}for every $k >0$, there exists a unique $x$ such that $\ell(x) = kx$, moreover, if $k_1> k_2$, then the corresponding $x_1, x_2$ satisfy $x_1 < x_2$.

 \noindent\hspace*{\dimexpr-\leftmargin\relax} If, in addition, $u$ satisfies Assumption \eqref{as:3}, then:
 \item the function\label{curve:cx} $\ell$ is strictly convex;
 \item \label{curve:diff}the function $\ell$ is continuously differentiable;

 \item \label{curve:axis-tangent}the derivative of $\ell$ satisfies $\lim_{x\to 0^+} \ell'(x) = -\infty$ and $\lim_{x\to \infty} \ell'(x) = 0$, so that for each $p>0$, there exists a unique $x$ such that $\ell'(x) = -p$;
 \item \label{curve:price}the derivative of $\ell$ is related to the marginal exchange rate by
 $
\ell'(x) = -p(u,x, \ell(x)).
$
 \end{enumerate}
\end{proposition}
We see that parts \eqref{curve:cx} and \eqref{curve:price} of Proposition~\ref{prop:level-curve} give the monotonicity of $p$ along any fixed level curve, while Assumption \eqref{as:4} is a global condition which in general cannot be deduced from the other assumptions.

For any utility function $u$, the collection of level sets is denoted by
$
\mathcal L_u = \left\{ s_{u,c}: c\in u(D)\right\}.
$
When $u:\R^2_+ \to \R_+$ satisfies Assumptions \eqref{as:1} and \eqref{as:2}, we will denote by $\ell_{u,c}$ the function constructed such that $\ell^g_{u,c} = s_{u,c}$. In this case, the level sets are referred to as level curves. In the context of $u$ being the trading function of a CFMM, the level curves are called liquidity curves. Let $\mathcal L_u^+ = \{s_{u,c}: c\in \R_{++}\}$ so that $\mathcal L_u = \mathcal L_u^+ \cup \{s_{u,0}\}$. Thus
$
\mathcal L_u^+ = \{\ell_{u,c}^g: c\in \R_{++}\}.
$
It will be convenient to keep the corresponding family of graphing functions separate and define
\[
\Lambda_u:=\{\ell_{u,c}:c\in\R_{++}\}.
\]
Clearly, if $u$ and $v$ are equivalent, we have $\mathcal L_u = \mathcal L_v$. The next proposition provides a partial converse.

\begin{proposition}\label{prop:equivalent-utility}
Let $u:\R_+^2\to\R_+$ satisfy Assumptions~\eqref{as:1}--\eqref{as:3},
and let $v:\R_+^2\to\R$ be a  utility function satisfying Assumption~\eqref{as:1} with
$\mathcal L_u=\mathcal L_v$. Then $u$ and $v$ are equivalent: for every
$\mathbf r^\alpha,\mathbf r^\beta\in\R_+^2$,
\[
u(\mathbf r^\alpha)\geq u(\mathbf r^\beta)
\quad\Longleftrightarrow\quad
v(\mathbf r^\alpha)\geq v(\mathbf r^\beta).
\]
\end{proposition}

When the level sets of $u$ are of the form $s_{u,c} = \ell^g_{u,c}$ and $v$ is equivalent to $u$, we see that the level sets of $v$ must also be of the form $\ell^g_{v,d}$. Any properties of $u$ that can be described by $\mathcal L_u$ are also satisfied by $v$.

 A utility function $u$ is \emph{homogeneous of degree one} if $u(kx,ky)=k u(x,y)$ for all $x,y,k>0$.  We will show that this is equivalent to the level curves, viewed as subsets of $\R^2$, being dilations of one another.
\begin{proposition}\label{prop:homothetic-preference}
 Let $u: \R_+^2 \to \R_+$ be a utility function satisfying Assumptions \eqref{as:1}, \eqref{as:2} and \eqref{as:3}. The following are equivalent:
 \begin{enumerate}[(i)]
 \item there exists a homogeneous utility function $v$ equivalent to $u$;
 \item for each $c, c'>0$ there exists $k> 0$ such that $\ell^g_{c} = k\ell^g_{c'}$, and $k > 1$ if and only if $c> c'$;
 \item the prices satisfy
 $p(u,x,y) = p(u,kx, ky)$ for any $x,y,k>0$.
 \end{enumerate}
 When the above holds,
 \begin{enumerate}[(a)]
 \item \label{homo:1}Assumption \eqref{as:4} holds for $u$;
 \item \label{homo:2}there is a strictly increasing function $g: \R_{++} \to \R_{++}$ such that
 $p(u,x,y)=g(y/x)$.
 \end{enumerate}
\end{proposition}
In such cases, we say that the utility function $u$, and also any utility function equivalent to $u$, represents a \textit{homothetic} preference. The function $g$ is called the \textit{ratio-price function}. Economically, this means that scaling both holdings by the same factor does not change the marginal rate of substitution. More specifically, willingness to exchange $X$ for $Y$ depends only on the ratio of holdings $y/x$, not on the scale of the portfolio. This scale invariance is the key aggregation property used in Section \ref{sec:cba}.

\begin{example}
The following are standard homothetic preferences; the corresponding
functions $v$ and $g$ are as in
Proposition~\ref{prop:homothetic-preference}.
\begin{enumerate}[(i)]
\item For $u(x,y)=\log x+\log y$, an equivalent representation satisfying
Assumptions~\eqref{as:1}--\eqref{as:4} is
$\varphi(x,y)=(xy)^{1/3}$. Moreover,
\[
v(x,y)=(xy)^{1/2},\qquad p(u,x,y)=\frac{y}{x},
\]
so $g(y/x)=y/x$.

\item For $u(x,y)=x^a y^b$, where $a,b>0$ and
$a+b<1$,
\[
v(x,y)=x^{a/(a+b)}y^{b/(a+b)},
\qquad
p(u,x,y)=\frac{a}{b}\frac{y}{x},
\]
so $g(y/x)=(a/b)(y/x)$.

\item For $u(x,y)=\bigl(xy/(x+y)\bigr)^\theta$, where
$\theta\in(0,1)$,
\[
p(u,x,y)=\left(\frac{y}{x}\right)^2,
\]
so $g(y/x)=(y/x)^2$. This utility satisfies Assumption~\eqref{as:1}
after continuous extension to the axes, but fails Assumption~\eqref{as:2}.
\end{enumerate}
\end{example}

In the numerical examples we sometimes write
\(u(x,y)=\log x+\log y\) because it makes utility changes easy to report.
On \(\R_{++}^2\), this utility represents the same preferences and the same
level curves as \((xy)^{1/3}\), which satisfies
Assumptions~\eqref{as:1}--\eqref{as:4} after continuous extension to
\(\R_+^2\). Whenever those assumptions are invoked for a logarithmic
constant-product example, the equivalent representation \((xy)^{1/3}\) is
understood. Numerical utility differences are reported under the logarithmic
normalization and are not representation-invariant welfare magnitudes; the
preference ordering of the reported outcomes is the economically relevant
conclusion.

Because liquidity provision is not modeled, only the set $\varphi^{-1}(\varphi(\mathbf r_0)) \subseteq \R_{+}^2$ is relevant. As shown by \citet{ACDEK2023}, with the invariant level $c$ fixed, any CFMM with trading function satisfying Assumptions~\eqref{as:1}--\eqref{as:3} can be represented by a homogeneous trading function $u$, by taking all positive two-dimensional dilations of the fixed liquidity curve. In terms of level-curve graphs, this means
$
\mathcal L_u^+=\{k\ell^g_{\varphi,c}: k>0\}.
$
Equivalently, in graphing-function notation,
\[
\Lambda_u=\left\{x\mapsto k\ell_{\varphi,c}\!\left(\frac{x}{k}\right): k>0\right\}.
\]

\section{Static equilibria}\label{sec:eq0}

We consider a closed market in which trading takes place exclusively through a CFMM mechanism. There is no external market where the assets can be exchanged, and agents can only interact with the CFMM, which is equipped with a trading function and a given inventory. Under this setting, we interpret $p(\varphi,x,y)$ as the CFMM price. The presence of an external market would allow agents to trade against outside prices, which is a separate problem that we do not consider in this paper.

A central feature of this environment is that prices are endogenous to the trading mechanism. Since a trade directly changes the CFMM inventory, it also changes the slope of the liquidity curve and hence the subsequent CFMM price. Therefore, it is no longer appropriate to model agents as facing an exogenously fixed price vector in the same way as in the classical Walrasian framework. Nevertheless, we are still interested in an equilibrium notion under which no agent has an incentive to trade further against the CFMM. This leads naturally to a unilateral no-trade notion of market equilibrium.

\subsection{Market states}\label{sec:market-notation}

Consider a market with two assets, $X$ and $Y$. The fee-free CFMM has trading function $\varphi$, and the two agents have utilities
$u_i:\R_+^2\to\R$, $i\in\{1,2\}$. The fee-free assumption is a benchmark
used to isolate the inventory mechanism.

A market state is a vector 
\[
\mathbf R=(\mathbf r_0,\mathbf r_1,\mathbf r_2)\in(\R_+^2)^3.
\]
It is \emph{interior} if $\mathbf R\in(\R_{++}^2)^3$. At an interior
state, write $p_0=p(\varphi,\mathbf r_0)$ and
$p_i=p(u_i,\mathbf r_i)$ for $i=1,2$. 
Throughout the paper, we use superscripts to distinguish market states. In particular, $\mathbf R^{\mathrm{in}}$ denotes the initial market state. When a state is fixed for a discussion, its superscript may be omitted.
For grouped holdings, let
\[
\mathbf r_A^\alpha:=\mathbf r_1^\alpha+\mathbf r_2^\alpha,\quad
\mathbf r_T^\alpha:=\mathbf r_0^\alpha+\mathbf r_A^\alpha,\quad
\mathbf r_{0i}^\alpha:=\mathbf r_0^\alpha+\mathbf r_i^\alpha.
\]
Here $i=1,2$, $A$ denotes all agents, $T$ the entire system, and $0i$ the CFMM together with agent $i$.

\begin{definition}
Given $\mathbf R^{\mathrm{in}}$, a market state $\mathbf R^\alpha$ is
\emph{feasible} if
\begin{enumerate}[(i)]
\item $\mathbf r_T^\alpha=\mathbf r_T^{\mathrm{in}}$;
\item $\varphi(\mathbf r_0^\alpha)=\varphi(\mathbf r_0^{\mathrm{in}})$.
\end{enumerate}
Write $\mathcal F(\mathbf R^{\mathrm{in}})$ for the set of feasible market states.
\end{definition}

Condition (i) is physical conservation of the two assets in the closed economy, while condition (ii) is conservation of the CFMM invariant. This notion of feasibility differs from that in a classical pure exchange economy because the agents' aggregate holdings $\mathbf r_A^\alpha$ may change. The CFMM is an inventory-bearing participant rather than a passive price vector.

On $\mathcal F(\mathbf R^{\mathrm{in}})$, define the preorder
$\mathbf R^\alpha\precsim\mathbf R^\beta$ if and only if
$u_i(\mathbf r_i^\alpha)\leq u_i(\mathbf r_i^\beta)$
for $i=1,2$.
Write $\mathbf R^\alpha\prec\mathbf R^\beta$ when at least one inequality
is strict. The feasible-improvement set is
\[
\mathcal U(\mathbf R^{\mathrm{in}})
:=\{\mathbf R^\alpha\in\mathcal F(\mathbf R^{\mathrm{in}}):
\mathbf R^{\mathrm{in}}\precsim\mathbf R^\alpha\}.
\]

\subsection{Motivating example}
We present a concrete example within our setting to illustrate how market states evolve. The insights from this example will be formalized and generalized later.

\begin{example}\label{ex:2}
Let $\varphi=u_1=u_2=\log x+\log y$ and
$\mathbf R^{\mathrm{in}}
=((100,100),(9,1),(1,9))$.
We consider a trading sequence in which the two agents take turns trading with the CFMM. At each step, the active agent chooses a trade that maximizes their utility given the current CFMM inventory. At each step $t\in \N$,
the market state is denoted by $\mathbf{R}^t=(\mathbf{r}_0^t,\mathbf{r}_1^t,\mathbf{r}_2^t) $
and the indifference price of agent $i$ is denoted by
$p_i^t=p(u_i,\mathbf{r}^t_i)$ for $i=0,1,2$, where $u_0=\varphi$.
The resulting sequence of market states is summarized in Table~\ref{tab:example-1}.
Theorem~\ref{thm:max-reach-limit} below proves that the sequence has a limit, which we denote by $\mathbf R^\infty$. Numerical computation gives, with coordinates rounded to three decimal places,
\begin{equation*}
 \mathbf R^\infty\approx\big((100,100),(4.928,4.928),(5.072,5.072)\big).
\end{equation*}
This implies $\lim_{t\rightarrow\infty} p_i^t =1$ for $i\in\{0,1,2\}$.
The pure exchange economy with the same agents, preferences, and initial holdings admits a unique Walrasian equilibrium given by
\[
\mathbf r_1^* = \mathbf r_2^* = (5,5),
\qquad
\mathbf q^* = (0.5,0.5).
\]
Hence, the limiting relative price agrees with the Walrasian equilibrium price, but the limiting allocation of the agents differs from the Walrasian equilibrium allocation.
Indeed, the allocation $\mathbf r_1^*=\mathbf r_2^*=(5,5)$ can be reached by a trading sequence if the two agents cooperate. The trading sequence is given in Table \ref{tab:market_states}.
\end{example}
\begin{table}[ht]
\centering
\begin{tabular}{c c c c}
\toprule
$t$ & $\mathbf r_0^t$ & $\mathbf r_1^t$ & $\mathbf r_2^t$ \\
\midrule
0 & $(100.000,\,100.000)$ & $(9.000,\,1.000)$ & $(1.000,\,9.000)$ \\
1 & $(103.885,\,96.260)$ & $(5.115,\,4.740)$ & $(1.000,\,9.000)$ \\
2 & $(99.822,\,100.179)$ & $(5.115,\,4.740)$ & $(5.063,\,5.082)$ \\
3 & $(100.009,\,99.991)$ & $(4.928,\,4.927)$ & $(5.063,\,5.082)$ \\
4 & $(100.000,\,100.000)$ & $(4.928,\,4.927)$ & $(5.072,\,5.073)$ \\
\bottomrule
\end{tabular}
\caption{The trading sequence that alternates utility-maximizing trades}
\label{tab:example-1}
\end{table}
\begin{table}[htbp]
\centering
\begin{tabular}{c ccc}
\toprule
$t$ & $\mathbf r_0^t$ & $\mathbf r_1^t$ & $\mathbf r_2^t$ \\
\midrule
0 & $(100.000,\,100.000)$
 & $(9.000,\,1.000)$ & $(1.000,\,9.000)$\\
1 & $(98.020,\,102.020)$
 & $(9.000,\,1.000)$ & $(2.980,\,6.980)$\\
2 & $(102.020,\,98.020)$
 & $(5.000,\,5.000)$ & $(2.980,\,6.980)$ \\
3 & $(100.000,\,100.000)$
 & $(5.000,\,5.000)$
 & $(5.000,\,5.000)$ \\
\bottomrule
\end{tabular}
\caption{The trading sequence that reaches the Walrasian equilibrium}
\label{tab:market_states}
\end{table}

This example isolates endogenous price impact. The sequential best-response path converges to a state in which the CFMM price and both agents' marginal rates of substitution agree, but different trading sequences can produce different terminal allocations at the same terminal price. 
Because each trade changes the terms faced by the next trader, trading order may matter.

\subsection{Equilibria}\label{sec:eq}
In this subsection, we characterize market equilibria under the discrete-time CFMM mechanism and examine whether these equilibria are efficient.

A trade $\mathbf h\in\R^2$ is called \emph{valid} for a market state
$\mathbf R^\alpha$ and agent $i\in\{1,2\}$ if
\[
\mathbf r_0^\alpha+\mathbf h\in\R_+^2,\quad
\varphi(\mathbf r_0^\alpha+\mathbf h)=\varphi(\mathbf r_0^\alpha)\quad \text{and} \quad \mathbf r_i^\alpha-\mathbf h\in\R_+^2.
\]
It is called \emph{utility-improving} if, in addition,
$u_i(\mathbf r_i^\alpha-\mathbf h)\geq u_i(\mathbf r_i^\alpha)$.

At a fixed market state, regard the zero trade as each agent's status quo
action. A unilateral deviation by agent \(i\) is any valid trade against the
current CFMM inventory while the other agent does not trade. The equilibrium
notion below therefore asks whether the zero-trade profile admits a
profitable unilateral deviation; in our setting, we do not allow 
multiple simultaneous nonzero orders.

\begin{definition}\label{def:4}
A market state $\mathbf R^\alpha$ is called a \emph{unilateral no-trade
equilibrium}, or simply a \emph{unilateral equilibrium}, if for each agent $i$, no valid trade by agent $i$ against the CFMM can strictly
improve that agent's utility. If
$\mathbf R^\alpha\in\mathcal F(\mathbf R^{\mathrm{in}})$, we say that it
is a unilateral equilibrium for the initial state
$\mathbf R^{\mathrm{in}}$.
\end{definition}

The next proposition records the one-agent optimization problem that underlies the unilateral no-trade
condition. It shows that a single utility-maximizing trade against the CFMM is characterized
by equality between the CFMM price and the agent's indifference price.

\begin{proposition}\label{prop:single-trade}
Let $\varphi$ and $u$ satisfy Assumptions~\eqref{as:1}--\eqref{as:4}. Fix an interior CFMM
inventory $\mathbf r_0\in\mathbb R^2_{++}$ and an interior agent holding
$\mathbf r\in\mathbb R^2_{++}$.
\begin{enumerate}[(a)]
\item \label{trade:exist} The problem
\[
\max_{\mathbf h} u(\mathbf r-\mathbf h)
\quad\text{subject to } \mathbf h \text{ valid}
\]
has an interior solution $\mathbf h^*$. That is, $\mathbf r-\mathbf h^*$ and $\mathbf r_0 + \mathbf h^*$ are interior.

\item \label{trade:optimal-implies-price} A utility-maximizing trade $\mathbf h^*$ satisfies $\mathbf r_0 + \mathbf h^*, \mathbf r - \mathbf h^* \in \R^2_{++}$ and the price equality
$
p(\varphi, \mathbf r_0 + \mathbf h^*) = p(u, \mathbf r - \mathbf h^*).
$
\item \label{trade:unique} The utility-maximizing trade $\mathbf h^*$ is unique.
\item \label{trade:price-implies-optimal}If a valid trade $\mathbf h'$ satisfies $\mathbf r_0 + \mathbf h', \mathbf r - \mathbf h' \in \R^2_{++}$ and the price equality
$
p(\varphi, \mathbf r_0 + \mathbf h') = p(u, \mathbf r - \mathbf h'),
$ then $\mathbf h'$ is the unique
utility-maximizing trade.

\item \label{trade:zero}The zero trade is the unique utility-maximizing trade if and only if
$
p(\varphi,\mathbf r_0)=p(u,\mathbf r).
$
\item \label{trade:sign} Let $\mathbf h^*$ be the unique utility-maximizing trade. If $p(\varphi,\mathbf r_0)>p(u,\mathbf r)$, then
$\operatorname{sgn}(\mathbf h^*)=(1,-1)$; if
$p(\varphi,\mathbf r_0)<p(u,\mathbf r)$, then
$\operatorname{sgn}(\mathbf h^*)=(-1,1)$.
\end{enumerate}
\end{proposition}

The unique trade in Proposition~\ref{prop:single-trade} is called \emph{utility-maximizing}. Part~\eqref{trade:sign} has a direct economic interpretation. If the CFMM price of $X$ exceeds the agent's marginal willingness to pay, $X$ is relatively expensive at the pool and the agent optimally sells $X$ to the CFMM; if the CFMM price is lower, the agent buys $X$. Along either best response, the CFMM price and the agent's indifference price move toward one another until they coincide. The next theorem gives an equivalent characterization of the unilateral equilibrium as stated in Definition~\ref{def:4}.

\begin{theorem}\label{thm:1}
Let $\varphi,u_1,u_2$ satisfy Assumptions~\eqref{as:1}--\eqref{as:4}.
An interior market state $\mathbf R^\alpha$ is a unilateral equilibrium
if and only if
\[
p(\varphi,\mathbf r_0^\alpha)=p(u_1,\mathbf r_1^\alpha)=p(u_2,\mathbf r_2^\alpha).
\]
\end{theorem}

Theorem~\ref{thm:1} shows that, at a unilateral equilibrium, all agents share the same indifference price, which coincides with the CFMM price. Indeed, Definition \ref{def:4} and Theorem \ref{thm:1} can be extended naturally to any finite number of agents, but we confine our discussions to the case of two trading agents for better geometric intuition. 
The economic content of the equality of prices is different from the classic Walrasian equilibrium. It says that each agent's indifference curve is tangent to the feasible frontier induced by the CFMM liquidity curve, so no single agent can improve by moving the pool along that curve. It does not say that agents face a common linear budget set or that their aggregate holdings clear independently of the CFMM inventory. This motivates the following definition.

\begin{definition}\label{def:5}
For an interior unilateral equilibrium
$\mathbf R^*\in\mathcal F(\mathbf R^{\mathrm{in}})$, the common value of the CFMM price $p_0$ and the agents' indifference prices $p_1, p_2$ is called the \emph{unilateral equilibrium price}.
\end{definition}

Suppose that $\mathbf R^\alpha,\mathbf R^\beta\in
\mathcal F(\mathbf R^{\mathrm{in}})$ are both interior unilateral
equilibria. A natural question is whether their unilateral equilibrium
prices necessarily coincide, that is, whether $p_0^\alpha=p_0^\beta$.
The following example shows that this is not the case.

\begin{example}\label{ex:3}
Let $\varphi(x,y)=u_2(x,y)=(xy)^{1/4}$ and
$u_1(x,y)=x^{1/8}y^{1/2}$, which satisfy
Assumptions~\eqref{as:1}--\eqref{as:4}.
Their marginal prices are
$p(\varphi,x,y)=p(u_2,x,y)=y/x$
and $p(u_1,x,y)=y/(4x)$.
Consider the two market states
\begin{align*}
    \mathbf R^\alpha & =
\left(
(2,1/2),\;(5/3,5/3),\;(1/3,1/12)
\right),\\
\mathbf R^\beta&=
\left(
(\sqrt 2,1/\sqrt 2),\;(1/6,1/3),\;(23/6-\sqrt 2,\;23/12-\sqrt 2/2)
\right).
\end{align*}
They have the same total holdings
$\mathbf r_T^\alpha=\mathbf r_T^\beta=(4,9/4)$
and invariant level
$\varphi(\mathbf r_0^\alpha)
=\varphi(\mathbf r_0^\beta)=1$.
Both states are unilateral equilibria, with
common marginal prices $1/4$ and $1/2$,
respectively. Thus the same feasible set
contains two unilateral equilibria with
different prices, even though each agent
has a homothetic preference.
\end{example}

\begin{proposition}\label{prop:unique-price}
 Let \(u_1\) and \(u_2\) represent the same homothetic preference, and
suppose that \(\varphi,u_1,u_2\) satisfy Assumptions~\eqref{as:1}--\eqref{as:4}. Then, for any initial market state
\(\mathbf R^{\mathrm{in}}\), any two interior unilateral equilibria
\(\mathbf R^\alpha,\mathbf R^\beta\in\mathcal F(\mathbf R^{\mathrm{in}})\)
have the same unilateral equilibrium price.
\end{proposition}

Common homothetic preferences force every agent's equilibrium holdings to have the same asset ratio at a given common price, so the conservation of total inventory in the economy pins down that price even though the division of aggregate holdings across agents can remain nonunique. Example~\ref{ex:3} shows that homotheticity of each agent’s preference alone does not guarantee a unique unilateral equilibrium price.

As illustrated by Examples~\ref{ex:2} and~\ref{ex:3}, unilateral equilibria are generally not unique. Our primary interest lies in those unilateral equilibria that belong to the set $\mathcal{U}(\mathbf R^{\mathrm{in}})$, which encode individual rationality relative to the initial state. We therefore ask whether such equilibria are Pareto optimal.

\begin{definition}
A market state $\mathbf R^*$ is called \emph{Pareto optimal} if there
exists no $\mathbf R^\alpha\in\mathcal F(\mathbf R^*)$ such that
\[
u_i(\mathbf r_i^\alpha)\geq u_i(\mathbf r_i^*),\qquad i=1,2,
\]
with strict inequality for at least one agent.
\end{definition}

Pareto optimality is defined relative to the total
inventory and CFMM invariant. In particular, if
$\mathbf R^*\in\mathcal F(\mathbf R^{\mathrm{in}})$,
then
$\mathcal F(\mathbf R^*)=\mathcal F(\mathbf R^{\mathrm{in}})$.
When $\mathbf R^{\mathrm{in}}$ is interior, every
$\mathbf R^*\in\mathcal U(\mathbf R^{\mathrm{in}})$
is also interior, since the CFMM invariant and
both agents' utilities remain positive.
The Pareto comparison concerns only the two agents,
with the CFMM invariant held fixed.
We next work toward an analogue of the First Welfare Theorem.
\begin{lemma}
\label{lem:compact-feasible-improvement}
Let \(\varphi,u_1,u_2\) be continuous. For every initial market state $\mathbf R^{\mathrm{in}}$, the set
$\mathcal U(\mathbf R^{\mathrm{in}})$ is nonempty and compact.
\end{lemma}

\begin{proposition}\label{prop:2}
Let \(\varphi,u_1,u_2\) be continuous. Given an initial market state $\mathbf R^{\mathrm{in}}$, the following
statements hold:
\begin{enumerate}[(i)]
\item there exists at least one Pareto optimal market state in
$\mathcal U(\mathbf R^{\mathrm{in}})$;
\item every Pareto optimal market state in
$\mathcal F(\mathbf R^{\mathrm{in}})$ is a unilateral equilibrium.
\end{enumerate}
\end{proposition}

Moreover, under suitable conditions, every unilateral equilibrium in $\mathcal U(\mathbf R^{\mathrm{in}})$ is Pareto optimal.
Here $\mathcal U(\mathbf R^{\mathrm{in}})$ restricts which equilibria are regarded as individually rational; it does not restrict the Pareto comparison once the candidate state lies in $\mathcal U(\mathbf R^{\mathrm{in}})$, because any feasible state dominating such a candidate automatically belongs to the same feasible-improvement set.

\begin{theorem}\label{thm:unilateral-implies-Pareto}
 Let $\varphi, u_1, u_2$ satisfy Assumptions~\eqref{as:1}--\eqref{as:4}. For any interior initial state $\mathbf R^{\mathrm{in}}$, every unilateral equilibrium in $\mathcal U(\mathbf R^{\mathrm{in}})$ is Pareto optimal.
\end{theorem}

This result is analogous to, but not identical with, the First Welfare Theorem. Pareto comparisons are taken over market states that keep total inventory and the CFMM invariant fixed, and the CFMM inventory is part of the feasible state. 
Part (ii) of Proposition~\ref{prop:2} can be seen as an analogue of the Second Welfare Theorem in our specialized framework.

\begin{remark}
The welfare results above admit a standard three-agent interpretation.
Recall that the CFMM trading function $\varphi$ may be viewed as a utility
function for the CFMM. If we regard the CFMM as agent~0, with utility
$u_0=\varphi$, then a market state
$
    \mathbf R=(\mathbf r_0,\mathbf r_1,\mathbf r_2)
$
is an allocation in a three-agent economy. Feasibility in our model requires
both conservation of total inventory and
$
    u_0(\mathbf r_0)
    =\varphi(\mathbf r_0)
    =\varphi(\mathbf r_0^{\mathrm{in}}),
$
so the CFMM's utility is held constant across all feasible allocations.
Consequently, Pareto optimality in our sense is precisely Pareto optimality
among three-agent allocations conditional on agent~0's utility remaining
fixed. 
Thus Proposition \ref{prop:2} and Theorem
\ref{thm:unilateral-implies-Pareto} may also be viewed as existence and
characterization results for such constrained Pareto-optimal allocations,
connecting the CFMM equilibrium problem directly with the welfare questions
studied in general equilibrium theory. From this perspective, Theorem \ref{thm:unilateral-implies-Pareto} can be understood through a standard supporting-price argument. At a unilateral
equilibrium, the common marginal price determines a vector
$q=(p_0^*,1)$ such that each participant's utility gradient, including that
of the CFMM, is a positive multiple of $q$. Concavity and conservation of
total inventory then rule out a feasible Pareto improvement.
\end{remark}

\subsection{Informal discussion for shifted CFMMs}

The assumptions imposed above are convenient because they imply that each liquidity curve has
the coordinate axes as asymptotes and that the induced price ranges over all of $(0,\infty)$.
This excludes, at least in its literal form, shifted or truncated constant-product market makers,
which arise in stylized descriptions of concentrated liquidity such as the virtual-reserve formulation in \citet{AZSKR2021}. We do not develop a full boundary
theory in this paper. The purpose of this subsection is only to indicate how the interior
price-equality condition should be interpreted when the relevant liquidity curve is truncated.
None of the results in later sections relies on this informal discussion.

Consider, for example,
\[
\varphi_{a,b}(x,y)=\log(x+a)+\log(y+b),
\qquad a,b>0.
\]
For a fixed level $c$, the liquidity curve is
 $ 
(x+a)(y+b)=e^c,
 $
or equivalently
 $ 
y=\frac{e^c}{x+a}-b.
$ 
Because actual holdings must satisfy $x,y\geq 0$, this curve is a shifted and truncated
constant-product curve. In particular, it is defined only for
\[
0\leq x\leq \frac{e^c}{b}-a,
\]
provided that the right endpoint is positive. The induced price is
\[
p(\varphi_{a,b},x,y)
=
\frac{y+b}{x+a}.
\]
Along the liquidity curve this becomes
\[
p(\varphi_{a,b},x,y)
=
\frac{e^c}{(x+a)^2}.
\]
Hence the attainable prices lie in the bounded interval
 $ 
\left[
\frac{b^2}{e^c},
\frac{e^c}{a^2}
\right].
 $ 

This bounded price range is the main new feature. If a unilateral equilibrium occurs at an interior
point of the truncated liquidity curve, then the same argument as in Theorem~\ref{thm:1} gives the usual
price-equality condition
$
p(\varphi_{a,b},\mathbf r_0^*)
=
p(u_1,\mathbf r_1^*)
=
p(u_2,\mathbf r_2^*).
$
At an endpoint, however, the first-order price condition need not hold. The correct condition is instead that no valid trade along the truncated liquidity curve can strictly improve the active agent's utility. Thus boundary equilibria should be interpreted as constrained optima. Equivalently, for shifted CFMMs, the unilateral no-trade condition consists of price equality at interior equilibria and a no-improving-trade condition at boundary equilibria.

\section{Aggregation and representative agents}\label{sec:cba}

After characterizing static equilibria, we next ask whether the aggregate behavior of several agents can be represented by a single representative agent. This question is useful for two reasons. First, it clarifies when a multi-agent CFMM economy can be reduced to a one-agent problem. Second, it identifies the geometric obstruction to such a reduction. Aggregate frontiers are obtained by infimal convolution and need not form the level curves of a single utility function.

The question is related to the classical exact-aggregation problem initiated by \citeauthor{Gorman1953}~(\citeyear{Gorman1953}, \citeyear{Gorman1961}), but the behavioral object is different. Classical Gorman aggregation asks when market demand at a common price can be written as the demand of one consumer independently of the distribution of income. Here agents face a nonlinear, inventory-dependent CFMM feasible frontier, and strong representation asks whether one utility reproduces aggregate best responses across all initial pool states and equilibria. The common-homothetic-preference condition derived below is therefore a mechanism-specific aggregation criterion, not a restatement of the Gorman polar-form theorem.

\subsection{Geometric motivation}
We now turn from existence and efficiency to aggregation. In traditional financial markets, a natural question is whether there exists a representative (or artificial) agent whose behavior replicates that of the aggregate economy. We investigate an analogous question in the setting of a CFMM. The main difference from the Walrasian setting is that the agent does not choose from a budget set with a linear frontier; instead, valid CFMM trades generate a translated nonlinear liquidity curve. This makes the geometry of level curves central.

Consider first one agent trading against a CFMM. Fix an initial market state and write
$
\mathbf T=\mathbf r_0^{\mathrm{in}}+\mathbf r_1^{\mathrm{in}}$ and $c_0=\varphi(\mathbf r_0^{\mathrm{in}})$.
If the CFMM inventory is \(\mathbf r_0=(x_0,y_0)\), then \(y_0=\ell_{\varphi,c_0}(x_0)\). Hence the agent's possible post-trade holdings, written in the agent's coordinates \(\mathbf r_1=(x,y)\), are described by
\begin{equation}\label{eq:one_agent_feasible_curve}
 y=T_y-\ell_{\varphi,c_0}(T_x-x).
\end{equation}
We call this the (one-agent) \emph{feasible frontier}.
By Proposition~\ref{prop:level-curve}, \(\ell_{\varphi,c_0}\) is strictly decreasing and strictly convex, so the frontier in \eqref{eq:one_agent_feasible_curve} is strictly decreasing and strictly concave. A utility-maximizing trade is obtained by moving to higher level curves of the agent until a level curve is tangent to this feasible frontier. At an interior optimum, this tangency condition is precisely
\(p(\varphi,\mathbf r_0)=p(u_1,\mathbf r_1)\),
which is the one-agent version of Theorem~\ref{thm:1}.

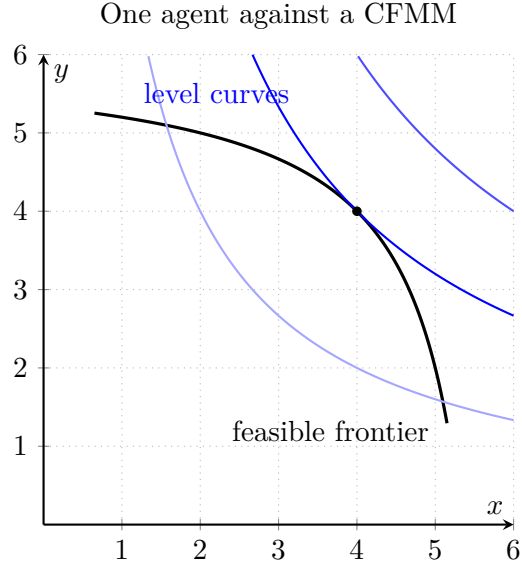
\begin{figure}[htbp]
\centering
\begin{tikzpicture}
\begin{axis}[
 width=7.8cm, height=7.8cm,
 xlabel={$x$}, ylabel={$y$},
 xmin=0, xmax=6,
 ymin=0, ymax=6,
 axis lines=middle,
 axis line style={thick, black},
 grid=major,
 grid style={dotted, gray!60},
 xtick={0,1,2,3,4,5,6},
 ytick={0,1,2,3,4,5,6},
 samples=300,
 smooth,
 no markers,
 restrict y to domain=0:6,
 title={One agent against a CFMM}
]
\addplot[black, very thick, domain=0.65:5.15] {6 - 4/(6-x)};
\addplot[blue!35, thick, domain=0.7:6] {8/x};
\addplot[blue, thick, domain=0.7:6] {16/x};
\addplot[blue!70, thick, domain=0.7:6] {24/x};
\fill (axis cs:4,4) circle[radius=1.8pt];
\node[anchor=north east] at (axis cs:5.05,1.45) {feasible frontier};
\node[anchor=south west, blue] at (axis cs:1.15,5.25) {level curves};

\end{axis}
\end{tikzpicture}
\caption{For a single agent, the liquidity curve induces a concave feasible frontier for the agent. The optimum occurs when the highest attainable level curve is tangent to the feasible frontier.}
\label{fig:one_agent_curve}
\end{figure}

For two agents, the same picture can be drawn in aggregate holdings. Fix
target utility levels \(c_1,c_2>0\). The set of holdings that give agent
\(i\) utility at least \(c_i\) is the epigraph
\(\ell_{u_i,c_i}^e\). Hence the aggregate holdings that can be split between
the two agents while giving utilities at least \(c_1\) and \(c_2\) form the
Minkowski sum
$
\ell_{u_1,c_1}^e+\ell_{u_2,c_2}^e.
$
The lower boundary of this set is the graph of the infimal convolution
(for standard convex-analytic background, see \citealp{HUL2001})
\begin{equation}\label{eq:inf_convolution}
(\ell_{u_1,c_1}\boxplus \ell_{u_2,c_2})(x)
:=
\inf_{x_1+x_2=x}
\left\{
\ell_{u_1,c_1}(x_1)+\ell_{u_2,c_2}(x_2)
\right\}.
\end{equation}

\begin{lemma}
\label{lem:infimal-convolution-regularity}
Let \(u_1,u_2\) satisfy Assumptions~\eqref{as:1}--\eqref{as:3}, fix
\(c_1,c_2>0\), and set
$
F:=\ell_{u_1,c_1}\boxplus\ell_{u_2,c_2}.
$
For every \(x>0\), the infimum defining \(F(x)\) is attained at a unique
interior split \(x=x_1(x)+x_2(x)\), with \(x_1(x),x_2(x)>0\). The function
\(F\) is continuously differentiable and strictly convex, and
$
F'(x)
=
\ell_{u_1,c_1}'(x_1(x))
=
\ell_{u_2,c_2}'(x_2(x)).
$
Consequently, if
\(\mathbf r_i=(x_i(x),\ell_{u_i,c_i}(x_i(x)))\), then
\(p(u_1,\mathbf r_1)=p(u_2,\mathbf r_2)=-F'(x)\).
\end{lemma}

For a utility function $u: \R_+^2 \to \R_+$ satisfying Assumptions \eqref{as:1}, \eqref{as:2} and \eqref{as:3}, write
$\mathcal L^e_{u}:= \{\ell^e_{u,c}: c >0\}$
for the collection of upper sets. The Minkowski sum of the epigraphs corresponds to inf-convolutions.
\begin{equation}
\mathcal L^e_{u_1}+\mathcal L^e_{u_2}
:= \left\{\ell^e_{u_1,c_1}+ \ell^e_{u_2,c_2}:c_1,c_2>0\right\}= \left\{\left(\ell_{u_1,c_1}\boxplus \ell_{u_2,c_2}\right)^e:c_1,c_2>0\right\}.
\end{equation}
The functions in
\begin{equation}\label{eq:aggregate_level_family}
 \Lambda_{u_1} \boxplus \Lambda_{u_2}:= \left\{\ell_{u_1,c_1}\boxplus \ell_{u_2,c_2}:c_1,c_2>0\right\}
\end{equation}
are called aggregate frontiers.
In general this encodes a two-parameter family of curves. It need not behave like the one-parameter family of level curves of a single utility function. The curves may cross, which is impossible for the level curves of one strictly increasing utility function. The strong representative-agent question is therefore equivalent to asking when the family in \eqref{eq:aggregate_level_family} collapses to an ordinary level-curve family.

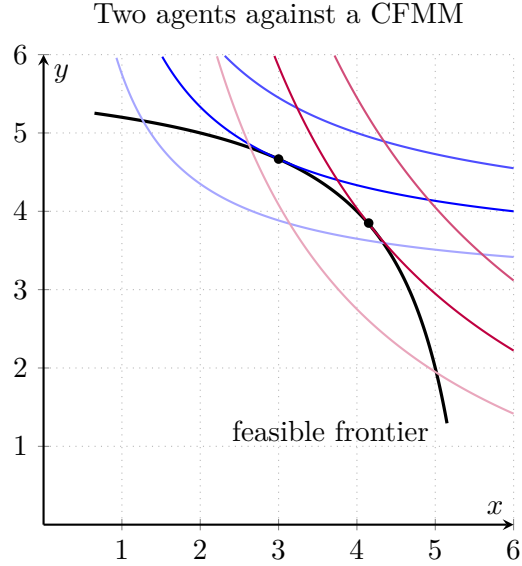
\begin{figure}[htbp]
\centering
\begin{tikzpicture}
\begin{axis}[
 width=7.8cm, height=7.8cm,
 xlabel={$x$}, ylabel={$y$},
 xmin=0, xmax=6,
 ymin=0, ymax=6,
 axis lines=middle,
 axis line style={thick, black},
 grid=major,
 grid style={dotted, gray!60},
 xtick={0,1,2,3,4,5,6},
 ytick={0,1,2,3,4,5,6},
 samples=300,
 smooth,
 no markers,
 restrict y to domain=0:6,
 title={Two agents against a CFMM}
]
\addplot[black, very thick, domain=0.65:5.15] {6 - 4/(6-x)};
\addplot[blue!35, thick, domain=0.7:6] {2.8/x + 2.95};
\addplot[blue, thick, domain=0.7:6] {4/x + 10/3};
\addplot[blue!70, thick, domain=0.7:6] {5.4/x + 3.65};
\addplot[purple!35, thick, domain=0.9:6] {16/x - 1.25};
\addplot[purple, thick, domain=0.9:6] {196/(9*x) - 38/27};
\addplot[purple!70, thick, domain=0.9:6] {28/x - 1.55};
\fill (axis cs:3,14/3) circle[radius=1.8pt];
\fill (axis cs:{21/5},{34/9}) circle[radius=1.8pt];

\node[anchor=north east] at (axis cs:5.05,1.45) {feasible frontier};

\end{axis}
\end{tikzpicture}
\caption{For two agents, aggregate frontiers are obtained by infimal convolution. Two different one-parameter subfamilies may each contain a curve tangent to the same feasible frontier, illustrating that aggregation can depend on which aggregate frontier is relevant.}
\label{fig:two_agent_noncollapse}
\end{figure}

\subsection{Weak representative agents}

We first record a weak, fixed-equilibrium notion of representative agent. This construction is useful because it works for arbitrary heterogeneous agents and arbitrary numbers of agents, although the resulting utility generally depends on the particular equilibrium being represented.

Fix a CFMM trading function \(\varphi\) and an \(n\)-agent economy with
utilities \(u_1,\ldots,u_n\). Let \([n]=\{1,\ldots,n\}\).
 For a market state $\mathbf R=(\mathbf r_0,\mathbf r_1,\ldots,\mathbf r_n)$,
write
$
\mathbf r_A:=\sum_{i=1}^n\mathbf r_i
$
for the agents' aggregate holding.
Let
\(\mathbf R^{\mathrm{in}}\in(\R_{++}^2)^{n+1}\) be an interior initial
market state, and let
\(\mathbf R^*\in\mathcal U(\mathbf R^{\mathrm{in}})\) be a unilateral equilibrium
with unilateral equilibrium price \(p_0^*\). The aggregate trade seen by the CFMM
is
$
\mathbf h^*:=\mathbf r_0^*-\mathbf r_0^{\mathrm{in}},
$
and feasibility gives
$
\mathbf r_A^*=\mathbf r_A^{\mathrm{in}}-\mathbf h^*.
$

A utility function \(u_A\) is called a \emph{weak representative agent} for
\((\mathbf R^{\mathrm{in}},\mathbf R^*)\) if \(\mathbf h^*\) uniquely solves
\[
\max_{\mathbf h}\;u_A(\mathbf r_A^{\mathrm{in}}-\mathbf h)
\]
over all valid trades \(\mathbf h\in\R^2\).
Equivalently, starting from the single-agent state
\((\mathbf r_0^{\mathrm{in}},\mathbf r_A^{\mathrm{in}})\), the agent with
utility \(u_A\) has the unique optimal post-trade state
\((\mathbf r_0^*,\mathbf r_A^*)\).

A utility function \(u_A\) is called a \emph{strong representative agent}
for \(u_1,\ldots,u_n\) relative to \(\varphi\) if the same \(u_A\),
independent of \(\mathbf R^{\mathrm{in}}\) and \(\mathbf R^*\), is a weak
representative agent for every interior initial state and every unilateral
equilibrium in \(\mathcal U(\mathbf R^{\mathrm{in}})\).

\begin{lemma}\label{lem:weighted-sup-convolution}
Let \(u_1,\ldots,u_n:\R_+^2\to\R\) be continuous and concave, and let
\(w_i>0\). Define
\[
U^w(\mathbf r)
:=
\sup\left\{
\sum_{i=1}^n w_i u_i(\mathbf r_i):
\mathbf r_i\in\R_+^2,\
\sum_{i=1}^n\mathbf r_i=\mathbf r
\right\}.
\]
Then the supremum is finite and attained for every \(\mathbf r\in\R_+^2\),
and \(U^w\) is continuous and concave. If each \(u_i\) is increasing, then
\(U^w\) is increasing. If, in addition, each \(u_i\) satisfies
Assumption~\eqref{as:1}, then \(U^w\) vanishes on the coordinate axes. Moreover, suppose each $u_i$ is differentiable at
\(\mathbf r_i^*\in\R_{++}^2\), put
\(\mathbf r_A^*=\sum_i\mathbf r_i^*\), and assume that
\(w_i\nabla u_i(\mathbf r_i^*)=\boldsymbol\lambda\) for every \(i\). Then
\((\mathbf r_1^*,\ldots,\mathbf r_n^*)\) solves the allocation problem at
\(\mathbf r_A^*\), and \(\boldsymbol\lambda\) is a supergradient of \(U^w\), so that
\begin{equation}\label{eq:weighted-sup-supergradient}
U^w(\mathbf r)
\leq
U^w(\mathbf r_A^*)
+\boldsymbol\lambda\cdot(\mathbf r-\mathbf r_A^*),
\qquad \mathbf r\in\R_+^2.
\end{equation}
\end{lemma}

\begin{proposition}\label{prop:weak-rep}
Let \(\mathbf R^{\mathrm{in}}\) be an interior initial market state in the
\(n\)-agent economy. Let
\(\mathbf R^*\in\mathcal U(\mathbf R^{\mathrm{in}})\)
be a unilateral equilibrium with unilateral equilibrium price \(p_0^*\).
Suppose that \(\varphi\) and \(u_i\) satisfy
Assumptions~\eqref{as:1}--\eqref{as:4} for each \(i\in[n]\).
Define weights
$
 w_i:=1/{\partial_y u_i(\mathbf r_i^*)}>0.
$
For \(\mathbf r_A\in\R_+^2\), define
\begin{equation}\label{eq:sup_convolution_rep}
 u_A^w(\mathbf r_A)
:=\sup\left\{\sum_{i=1}^n w_i u_i(\mathbf r_i):
 \sum_{i=1}^n \mathbf r_i=\mathbf r_A,\; \mathbf r_i\in\R_+^2\right\},
 \qquad \mathbf r_A\in\R_+^2,
\end{equation}
Then \(u_A^w\) is a weak representative agent for
\((\mathbf R^{\mathrm{in}},\mathbf R^*)\).
\end{proposition}

\begin{remark}
Proposition~\ref{prop:weak-rep} is weaker than the state-independent results below, but it is still useful for two reasons. First, it applies to an arbitrary number of heterogeneous agents and to an arbitrary fixed unilateral equilibrium. Second, it gives a planner interpretation. For a fixed aggregate holding \(\mathbf r_A\), the representative utility \(u_A^w(\mathbf r_A)\) is the maximal weighted total utility attainable by splitting \(\mathbf r_A\) among the agents. The limitation is that the weights depend on \(\mathbf R^*\), so this construction is generally not a strong representative agent. The normalization $w_i=1/\partial_y u_i(\mathbf r_i^*)$ makes every weighted marginal-utility vector equal to $(p_0^*,1)$. Thus the weights are precisely local planner weights supporting the equilibrium allocation of the aggregate bundle; their dependence on equilibrium marginal utilities explains why weak aggregation need not extend across states.
\end{remark}

\begin{remark}\label{rmk:2}
The weighted sup-convolution in \eqref{eq:sup_convolution_rep} is the demand-side analogue of a cost-minimization problem for combined market makers. It maximizes weighted utility over all allocations of a fixed aggregate holding. Dually, one can allocate a fixed trade or inventory across multiple liquidity sources so as to minimize cost. This is closely related to the geometric viewpoint in \citet{ACDEK2023} and to the infimal-convolution construction of combined cost functions in \citet{BFP25}. The important distinction is that Proposition~\ref{prop:weak-rep} is local to a chosen equilibrium through the weights \(w_i\), while the strong results below characterize when no such state-dependent weights are needed.
\end{remark}

\subsection{Strong representative agents and collapse of aggregate curves}

For the rest of this section, we return to the two-agent economy and ask when there is a single utility function that represents the two agents for every initial market state.

The following results form the convex-geometric core. They show that the two-parameter family of aggregate frontiers collapses to a one-parameter level-curve family exactly when the agents share the same homothetic preference. For a curve \(f\) satisfying the conclusions of Proposition~\ref{prop:level-curve}, define its tangent-intercept transform by
\[
 A_f(q):=\inf_{x>0}\{qx+f(x)\}, \qquad q>0.
\]

\begin{lemma}\label{lem:cont-of-transform}
Let \(u\) satisfy Assumptions~\eqref{as:1}--\eqref{as:4}. For each fixed
\(q_0>0\), the map
\[
\alpha_u:(0,\infty)\to(0,\infty),
\qquad
\alpha_u(c):=A_{\ell_{u,c}}(q_0),
\]
is continuous, strictly increasing, and onto. In particular, for every
\(t>0\), there is a unique level curve \(\ell_{u,c}\) satisfying
\(A_{\ell_{u,c}}(q_0)=t\).
\end{lemma}

\begin{lemma}\label{lem:collapse}
Let \(u_1,u_2,v\) satisfy Assumptions~\eqref{as:1}--\eqref{as:4}. Then
\begin{equation}\label{eq:collapse_statement}
 \Lambda_{u_1} \boxplus \Lambda_{u_2}=\Lambda_v
\end{equation}
if and only if \(u_1,u_2\), and \(v\) represent the same homothetic preference.
\end{lemma}

\begin{remark}\label{rmk:3}
The transform \(A_f(q)=\inf_x\{qx+f(x)\}\) can be interpreted as the minimum value of the linear cost \(qx+y\) required to reach the graph of \(f\). The identity \(A_{f\boxplus g}=A_f+A_g\) says that aggregation of agents by Minkowski summation becomes addition of these dual cost functions. Lemma~\ref{lem:collapse} therefore has a cost-minimization interpretation. The only way the dual family can be closed under addition while remaining one-dimensional is for all level curves to be positive dilations of one another. This is the demand-side counterpart of infimal-convolution constructions for combined market makers.
\end{remark}

Next, we characterize when a strong representative agent exists.

\begin{theorem}\label{thm:strong-rep}
Fix a CFMM trading function \(\varphi\) satisfying Assumptions~\eqref{as:1}--\eqref{as:4}. Let \(u_1,u_2,v\) satisfy Assumptions~\eqref{as:1}--\eqref{as:4}. Then the following are equivalent:
\begin{enumerate}[(i)]
 \item \(v\) is a strong representative agent for \(u_1,u_2\);
 \item \(\Lambda_{u_1} \boxplus \Lambda_{u_2}=\Lambda_v\);
 \item \(u_1,u_2\), and \(v\) represent the same homothetic preference.
\end{enumerate}
\end{theorem}

The economic force behind Theorem~\ref{thm:strong-rep} is the absence of scale-dependent composition effects. When all agents share one homothetic preference, a common marginal price determines the same desired $Y/X$ ratio for every agent, regardless of the size of the agent's position. Aggregate holdings therefore inherit that ratio. With heterogeneous ratio-price functions, the distribution of holdings across agents remains an economically relevant state variable, so no state-independent utility satisfying our assumptions can reproduce all aggregate equilibrium holdings.

\begin{remark}
Theorem~\ref{thm:strong-rep} is stated for representative agents, but the same geometric obstruction appears when liquidity sources are combined. A one-time pooled frontier can be formed by the corresponding Minkowski sum or infimal convolution. The stronger requirement is that the entire two-parameter family of pooled frontiers again be representable as the one-parameter level-curve family of a single CFMM, so that pooling liquidity sources is compatible with ordinary state-independent liquidity provision to the pooled market. This is exactly the same collapse condition as $\Lambda_{u_1}\boxplus\Lambda_{u_2}=\Lambda_v$. Hence, under the same regularity assumptions, such strong pooling is possible only when the component trading functions represent the same homothetic preference. Heterogeneous liquidity curves can still be combined locally or routed across, but the combined object need not be equivalent to a single CFMM with a state-independent liquidity family.
\end{remark}

\begin{remark}\label{rmk:4}
The obstruction can be seen directly. Suppose two aggregate frontiers \(F,G\in\Lambda_{u_1} \boxplus \Lambda_{u_2}\) cross at the same aggregate holding \(\mathbf r_A\), but have different tangent prices there. The split generating \(F\) can be made into an initial market state that is already a unilateral equilibrium by choosing a CFMM inventory with the same price. Strong representation would force \(p(v,\mathbf r_A)\) to equal that tangent price. Doing the same with the split generating \(G\) would force \(p(v,\mathbf r_A)\) to equal the other tangent price. Since a differentiable utility has only one indifference price at \(\mathbf r_A\), this is impossible. Thus non-collapsing geometry is exactly the failure of state-independent representative behavior.
\end{remark}

\begin{example}\label{ex:4}
Suppose \(u_i(x,y)=\log x+\log y\) for \(i\in[n]\). Then
\(p(u_i,x,y)=y/x\). If the agents have a common unilateral equilibrium price
\(p^*\), then \(y_i^*/x_i^*=p^*\) for every \(i\in[n]\). Summing over the
agents gives
$
{y_A^*}/{x_A^*}=p^*.
$
Hence the aggregate utility \(u_A(x,y)=\log x+\log y\) has the same
indifference price at the aggregate allocation. This is the classical
constant-product case in which all agents share the same homothetic preference.
\end{example}

\section{Reachability and trading dynamics}\label{sec:tr}

The preceding sections characterize equilibrium states. We now ask whether unilateral equilibria in $\mathcal{U}(\mathbf{R}^{\mathrm{in}})$ can be reached through a trading sequence, and how trading order affects welfare. Throughout this section, we continue to focus on the two-agent setting.

\subsection{Trading sequences and reachability}

The trading sequence can be characterized as a sequence of market states, defined as follows.

\begin{definition}\label{def:7}
A sequence of market states $(\mathbf R^t)_{t=0}^T$ with $T\in\mathbb N\cup\{\infty\}$ is called a \emph{trading sequence} if for every $t<T$ there exists a valid trade $\mathbf h^t \in \mathbb R^2$, executed by one of the agents $i\in\{1,2\}$, such that
\[
\mathbf r_0^{t+1} = \mathbf r_0^{t} + \mathbf h^t,\quad
\mathbf r_i^{t+1} = \mathbf r_i^{t} - \mathbf h^t,\quad
\mathbf r_j^{t+1} = \mathbf r_j^{t} \ \text{for } j \neq i.
\]
We abuse language slightly and say that the trading sequence consists of trades $\mathbf h^0, \mathbf h^1, \dots, \mathbf h^{T-1}$. If $T=\infty$ and the limit exists, we write $\mathbf R^T = \mathbf R^\infty = \lim_{t\to\infty}\mathbf R^t$.
\end{definition}

The index \(t\) counts individual trades. In an alternating sequence, a \emph{round} consists of two consecutive trades, one by each agent. We next define reachability from
an initial market state.
\begin{definition}
Let $\mathbf R^{\mathrm{in}}$ be an initial market state.
\begin{enumerate}[(i)]
\item A market state $\mathbf R^\alpha$ is said to be \emph{validly reachable} from $\mathbf R^{\mathrm{in}}$
if there exists a trading sequence of valid trades $(\mathbf R^t)_{t=0}^T$
with $\mathbf R^0=\mathbf R^{\mathrm{in}}$ and $\mathbf R^T=\mathbf R^\alpha$ where $T \in \mathbb N \cup \{\infty\}$. The collection of validly reachable states from $\mathbf R^{\mathrm{in}}$ is denoted by $\mathcal R^{\mathrm{val}}(\mathbf R^{\mathrm{in}})$.

\item A market state $\mathbf R^\alpha$ is said to be \emph{improvement reachable} from $\mathbf R^{\mathrm{in}}$
if there exists a sequence of utility-improving trades $(\mathbf R^t)_{t=0}^T$
with $\mathbf R^0=\mathbf R^{\mathrm{in}}$ and $\mathbf R^T=\mathbf R^\alpha$ where $T \in \mathbb N \cup \{\infty\}$. The collection of improvement reachable states from $\mathbf R^{\mathrm{in}}$ is denoted by $\mathcal R^{\mathrm{imp}}(\mathbf R^{\mathrm{in}})$.

\item A market state $\mathbf R^\alpha$ is said to be \emph{utility-maximization reachable} from $\mathbf R^{\mathrm{in}}$
if there exists a sequence of utility-maximizing trades $(\mathbf R^t)_{t=0}^T$
with $\mathbf R^0=\mathbf R^{\mathrm{in}}$ and $\mathbf R^T=\mathbf R^\alpha$ where $T \in \mathbb N \cup \{\infty\}$. The collection of utility-maximization reachable states from $\mathbf R^{\mathrm{in}}$ is denoted by $\mathcal R^{\mathrm{max}}(\mathbf R^{\mathrm{in}})$.
\end{enumerate}
\end{definition}

The distinction between these reachability notions has an economic interpretation. Valid reachability describes what the mechanism permits mechanically; improvement reachability imposes individual rationality along the path; utility-maximization reachability imposes myopic optimality at every trade. The three notions therefore separate technological feasibility from behavioral restrictions. We first record some trivial inclusions without proof.
\begin{proposition}
Let $\varphi,u_1,u_2$ be continuous. For any initial state $\mathbf R^{\mathrm{in}}$, we have
\begin{enumerate}[(i)]
 \item $\mathcal R^{\mathrm{max}}(\mathbf R^{\mathrm{in}}) \subseteq \mathcal R^{\mathrm{imp}}(\mathbf R^{\mathrm{in}}) \subseteq \mathcal R^{\mathrm{val}}(\mathbf R^{\mathrm{in}})$;
 \item $\mathcal R^{\mathrm{imp}}(\mathbf R^{\mathrm{in}}) \subseteq \mathcal U(\mathbf R^{\mathrm{in}})$;
 \item $\mathcal R^{\mathrm{val}}(\mathbf R^{\mathrm{in}})\subseteq \mathcal F(\mathbf R^{\mathrm{in}})$.
\end{enumerate}

\end{proposition}
The continuity assumption is needed only because the reachability definitions allow infinite trading sequences. For finite sequences, the inclusions follow directly from the definitions.

\begin{example}\label{ex:5}
We display some examples demonstrating the strictness of inclusions.
\begin{enumerate}[(i)]
\item We first exhibit a market state in $\mathcal R^{\mathrm{val}}(\mathbf R^{\mathrm{in}}) \setminus \mathcal U (\mathbf R^{\mathrm{in}}) $ with equal marginal exchange rates.
Let $\varphi = u_1 = u_2$ be the function $(x,y)\mapsto\log x + \log y$ and $
\mathbf R^{\mathrm{in}} = \big((1,1),(9,1),(1,9)\big)$.
The trading sequence in Table~\ref{tab:reachable-not-U}
ends at a state $\mathbf R^3$ with equal marginal prices,
but agent~1 has lower utility than initially. Hence
$
\mathbf R^3 \in \mathcal R^{\mathrm{val}}(\mathbf R^0)$
but
$
\mathbf R^3 \notin \mathcal U(\mathbf R^0).
$

\item \label{ex5pt2} $\mathcal F(\mathbf R^{\mathrm{in}}) \not \subseteq \mathcal R^{\mathrm{val}}(\mathbf R^{\mathrm{in}})$.
Let $\varphi=u_1=u_2$ be $(x,y)\mapsto(xy)^{1/3}$, and take
\[
\begin{aligned}
\mathbf R^{\mathrm{in}}
&=\big((100,100),(0,0),(10,10)\big),\\
\mathbf R^\alpha
&=\big((100,100),(5,5),(5,5)\big).
\end{aligned}
\]
The target $\mathbf R^\alpha$ is feasible. However, agent~1 starts with zero holdings and cannot execute any nonzero
valid trade. Indeed, $\mathbf r_1^{\mathrm{in}}-\mathbf h\geq0$ implies
$\mathbf h\leq0$, which together with strict monotonicity of $\varphi$
forces $\mathbf h=0$. Trades by agent~2 do not change agent~1's holdings. Hence agent~1 remains at $(0,0)$ after every finite trading sequence and also in the limit of every infinite trading sequence, so $\mathbf R^\alpha$ is not validly reachable.
\end{enumerate}
\end{example}
\begin{table}[ht]
\centering
\begin{tabular}{c c c c}
\toprule
$t$ & $\mathbf r_0^t$ & $\mathbf r_1^t$ & $\mathbf r_2^t$ \\
\midrule
$0$
&
$(1.000,\,1.000)$
&
$(9.000,\,1.000)$
&
$(1.000,\,9.000)$
\\
$1$
&
$(9.900,\,0.101)$
&
$(0.100,\,1.899)$
&
$(1.000,\,9.000)$
\\
$2$
&
$(2.245,\,0.446)$
&
$(0.100,\,1.899)$
&
$(8.655,\,8.655)$
\\
$3$
&
$(1.000,\,1.000)$
&
$(1.345,\,1.345)$
&
$(8.655,\,8.655)$
\\
\bottomrule
\end{tabular}
\caption{A trading sequence that reaches a market state outside $\mathcal U(\mathbf R^{\mathrm{in}})$.}
\label{tab:reachable-not-U}
\end{table}

Part \eqref{ex5pt2} of Example~\ref{ex:5} relies on using non-interior market states. The following theorem holds if we restrict to interior market states. 

\begin{theorem}
\label{thm:interior-feasible-reachable}
Let $\varphi$ satisfy Assumptions~\eqref{as:1}--\eqref{as:3}, and fix an
interior initial state $\mathbf R^{\mathrm{in}}$. Every interior state
$\mathbf R^\alpha\in\mathcal F(\mathbf R^{\mathrm{in}})$ is reachable
from $\mathbf R^{\mathrm{in}}$ by a finite valid trading sequence.
Moreover, the sequence can be chosen so that all intermediate states
are interior.
\end{theorem}

 The preceding reachability result does not depend on $u_1$ and $u_2$. In particular, it can extend to any finite number of agents. Any agent may play the role of the market maker as in the proof of Lemma~\ref{lem:local-valid-reachability} 
 while the local trading loops adjust the holdings of each of the remaining agents, and the interior feasible-state space remains path connected. Thus valid reachability is a property of the CFMM trading mechanism and its feasible-state geometry.  The next corollary shows that, if the market starts with the Walrasian equilibrium price, then the Walrasian equilibrium allocation is validly reachable.

\begin{corollary}\label{cor:classical-eq}
Let \(\varphi,u_1,u_2\) satisfy
Assumptions~\eqref{as:1}--\eqref{as:4}. Let
$
\mathbf R^{\mathrm{in}}
$
be an interior initial market state, and suppose that
\((\mathbf r_1^{\mathrm W},\mathbf r_2^{\mathrm W},\mathbf q^{\mathrm W})\) is a Walrasian equilibrium of
the pure exchange economy with initial endowments
\((\mathbf r_1^{\mathrm{in}},\mathbf r_2^{\mathrm{in}})\). 
Assume that $
p(\varphi,\mathbf r_0^{\mathrm{in}})=q_1^{\mathrm W}/q_2^{\mathrm W}.
$
Then
$
\mathbf R^{\mathrm W}
=
(\mathbf r_0^{\mathrm{in}},\mathbf r_1^{\mathrm W},\mathbf r_2^{\mathrm W})
$
is a unilateral equilibrium and is validly reachable from $\mathbf R^{\mathrm{in}}$.
\end{corollary}

Corollary~\ref{cor:classical-eq} is an implementation result, not a convergence result. It identifies a finite sequence of valid trades that implements the Walrasian allocation; the individual trades need not improve utility or be best responses. Initial price alignment matters because the target keeps the CFMM at its initial inventory while reallocating the agents' aggregate endowment.

The following example shows that an interior state in $\mathcal U(\mathbf R^{\mathrm{in}})$ may be validly reachable but not improvement reachable.

\begin{example}
\label{ex:interior-not-improvement-reachable}
Let $\varphi=u_1=u_2=(xy)^{1/3}$ and
$\mathbf z(a)=(a,a^{-1})$ for $a>0$.
Consider the two market states
\[
\begin{aligned}
\mathbf R^{\mathrm{in}}
&=(\mathbf z(1),\mathbf z(2),\mathbf z(6)),\\
\mathbf R^\alpha
&=\left(
\mathbf z(3),
\mathbf z\!\left(3+\frac{3}{\sqrt2}\right),
\mathbf z\!\left(3-\frac{3}{\sqrt2}\right)
\right).
\end{aligned}
\]
Both states are interior, have total inventory
$(9,5/3)$, and satisfy $xy=1$ for every
participant. Thus
$\mathbf R^\alpha\in\mathcal U(\mathbf R^{\mathrm{in}})$
and is validly reachable by
Theorem~\ref{thm:interior-feasible-reachable}.
Suppose an improvement trading sequence reaches
$\mathbf R^\alpha$. Since both agents' utilities
equal $1$ at the endpoints and cannot decrease,
they remain equal to $1$ throughout.
The CFMM invariant also remains fixed, so
every participant always satisfies $xy=1$.
For any trade, conservation of both assets implies
that the CFMM and the active agent have the same
sum and reciprocal sum of their $X$-holdings
before and after trading, hence the same product.
Their holdings can therefore only be preserved
or exchanged.
Consequently, the sequence can visit only
permutations of the initial three holdings,
including in the limit. None has CFMM
$X$-holding equal to $3$, so
$\mathbf R^\alpha$ is not improvement reachable.
\end{example}

\subsection{Alternating utility-maximizing trades}
We have seen that all interior feasible states, including those that are unilateral equilibria, are validly reachable. A complete characterization of improvement-reachable unilateral equilibria is generally nontrivial. We therefore turn to a constructive result. Under alternating utility-maximizing behavior, a utility-maximization reachable unilateral equilibrium exists and is selected by the trading process.

\begin{theorem}\label{thm:max-reach-limit}
Let $\varphi,u_1,u_2$ satisfy Assumptions~\eqref{as:1}--\eqref{as:4}.
Starting from an interior market state $\mathbf R^{\mathrm{in}}$, suppose that the two agents take turns trading with the CFMM and each active agent maximizes her own utility.
Then the sequence of market states converges to an interior state
$(\mathbf r_0^*,\mathbf r_1^*,\mathbf r_2^*)$, and
\[
p(\varphi,\mathbf r_0^*)
=p(u_1,\mathbf r_1^*)
=p(u_2,\mathbf r_2^*).
\]
\end{theorem}

\begin{corollary}
Under the assumptions of Theorem~\ref{thm:max-reach-limit}, the limiting market state is a unilateral equilibrium. Moreover, the limit belongs to
\(\mathcal U(\mathbf R^{\mathrm{in}})\) and is Pareto optimal.
\end{corollary}


This convergence is a best-response dynamic rather than Walrasian t\^atonnement. There is no auctioneer adjusting a price while allocations remain fixed. Every best response changes the pool's actual inventory, and therefore changes both the price and the feasible frontier faced by the next agent. Convergence follows from endogenous price bracketing together with monotone movements in the agents' holdings.

The preceding discussion also identifies the set of utility-maximization reachable market states. Once the initial state is fixed, the only
choice is which agent moves first and how many nonzero utility-maximizing
trades are performed. Indeed, by Proposition~\ref{prop:single-trade}
\eqref{trade:unique}, the active agent's utility-maximizing trade is
unique. Moreover, after an agent trades, the CFMM price equals that
agent's indifference price, so an immediate second trade by the same
agent is the zero trade by Proposition~\ref{prop:single-trade}
\eqref{trade:zero}. Hence, up to inserting zero trades, every
utility-maximizing trading sequence is one of the two alternating paths,
depending on whether agent~1 or agent~2 moves first, stopped after a
finite number of trades. Thus the utility-maximization reachable states are precisely the
states along these two alternating utility-maximizing trading sequences, including the limiting case as in Theorem~\ref{thm:max-reach-limit}.

\subsection{First-mover effects under alternating utility maximization}\label{sec:first-mover}

Throughout the subsection, a first-mover advantage or disadvantage compares the utility
of the \emph{same agent} under the two possible trading orders, holding the
initial market state fixed.

The next proposition shows that, given two inventories on the same liquidity curve whose CFMM prices lie on the same side of the agent's indifference price, the one farther from that price leads to a strictly higher post-trade utility for the agent.

\begin{proposition}\label{prop:first-mover}
Let $\varphi$ and $u_i$ satisfy Assumptions~\eqref{as:1}--\eqref{as:4}. Fix $i\in\{1,2\}$, an interior agent holding $\mathbf r_i\in\R_{++}^2$, and two interior CFMM inventories $\mathbf r_0,\widetilde{\mathbf r}_0\in\R_{++}^2$ on the same liquidity curve. 
Set $\widetilde p_0=p(\varphi,\widetilde{\mathbf r}_0)$.
Let $\mathbf r_i^*$ and $\widetilde{\mathbf r}_i^*$ be the agent's post-trade holdings after the respective utility-maximizing trades. If
$p_0>\widetilde p_0>p_i$ or $p_0<\widetilde p_0<p_i$, then
$u_i(\mathbf r_i^*)>u_i(\widetilde{\mathbf r}_i^*)$.
\end{proposition}

We now use Proposition~\ref{prop:first-mover}
to compare the two possible trading orders.
For $i\in\{1,2\}$, let \(B_i(\mathbf R^\alpha)\) denote the state obtained
from an interior state \(\mathbf R^\alpha\) after agent \(i\) makes her
unique utility-maximizing trade. From an initial state
\(\mathbf R^{\mathrm{in}}\), write \(\mathbf R^{t,a}\) for the alternating
utility-maximizing sequence whose first mover is \(a\in\{1,2\}\), where
\(t\) is the number of completed trades. For example,
\[
\begin{aligned}
\mathbf R^{0,a}
&=\mathbf R^{\mathrm{in}},
\qquad
\mathbf R^{1,a}
=B_a(\mathbf R^{\mathrm{in}}),\\
\mathbf R^{2,a}
&=B_{3-a}\bigl(B_a(\mathbf R^{\mathrm{in}})\bigr).
\end{aligned}
\]
The same superscripts are used for inventories, prices, and trades, such as
\(\mathbf r_i^{t,a}\), \(p_i^{t,a}\), and \(\mathbf h^{t,a}\).
When the sequence of market states converges, its limit is denoted by
\(\mathbf R^{\infty,a}\).
In the case where $p_1^0,p_2^0$ are on the same side of $p_0^0$, we say
that the \emph{first-trade no-crossing condition} holds if
\begin{equation}\label{con:no-crossing}
\begin{aligned}
(p_0^0-p_1^0)(p_0^{1,2}-p_1^0)&>0,\\
(p_0^0-p_2^0)(p_0^{1,1}-p_2^0)&>0.
\end{aligned}
\end{equation}

\begin{proposition}
\label{prop:first-mover-effect}
Let \(\varphi,u_1,u_2\) satisfy Assumptions~\eqref{as:1}--\eqref{as:4}.

\begin{enumerate}[(i)]
\item \label{eq:first-mover-effect1}If $(p_0^0-p_1^0)(p_0^0-p_2^0)>0$
and condition~\eqref{con:no-crossing} holds, then
\[
u_1(\mathbf r_1^{2,1})>u_1(\mathbf r_1^{2,2}),
\qquad
u_2(\mathbf r_2^{2,2})>u_2(\mathbf r_2^{2,1}).
\]
\item \label{eq:first-mover-effect2}If $(p_0^0-p_1^0)(p_0^0-p_2^0)<0$,
then
\[
u_1(\mathbf r_1^{2,1})<u_1(\mathbf r_1^{2,2}),
\qquad
u_2(\mathbf r_2^{2,2})<u_2(\mathbf r_2^{2,1}).
\]
\end{enumerate}
\end{proposition}
Proposition \ref{prop:first-mover-effect}
identifies two first-round order effects.
When both agents initially wish to trade in the
same direction, the first mover has an advantage
under the no-crossing condition.
When they wish to trade in opposite directions,
the first mover is disadvantaged. These comparisons are the closed-market analogue of transaction-order dependence in decentralized exchanges studied, for example, by \citet{XMP2023}: even without outside arbitrage, the first trade changes the state variable that determines the next trader's terms. 

We next strengthen Proposition~\ref{prop:first-mover-effect}~\eqref{eq:first-mover-effect1}
by showing that, in the same-side case, the first-mover
advantage is secured by the first trade and persists
throughout the alternating sequence, including at its limit.

\begin{theorem}
\label{thm:first-mover-advantage-limit}
Let \(\varphi,u_1,u_2\) satisfy Assumptions~\eqref{as:1}--\eqref{as:4}.
Suppose that $
(p_0^0-p_1^0)(p_0^0-p_2^0)>0$
and condition~\eqref{con:no-crossing} holds.
Then we have
\begin{equation}\label{eq:limiting-advantage-agent1}
u_1(\mathbf r_1^{\infty,1})
\geq u_1(\mathbf r_1^{1,1})
>u_1(\mathbf r_1^{\infty,2}),
\end{equation}
\begin{equation}\label{eq:limiting-advantage-agent2}
u_2(\mathbf r_2^{\infty,2})
\geq u_2(\mathbf r_2^{1,2})
>u_2(\mathbf r_2^{\infty,1}).
\end{equation}
\end{theorem}

We now return to the opposite-side case in
Proposition~\ref{prop:first-mover-effect}~\eqref{eq:first-mover-effect2}.
The following example illustrates the evolution
of the first-mover disadvantage under two
different initial CFMM reserves.

\begin{example}\label{ex:6}
Let $\varphi=u_1=u_2=\log x+\log y$.
Fix the agents' initial holdings at
$\mathbf r_1^0=(9,1)$ and $\mathbf r_2^0=(2,8)$.
We consider two initial CFMM inventories,
$(100,100)$ and $(3,6)$.
As shown in Figure~\ref{fig:util-alt_2},
the first-mover disadvantage persists
at the plotted late stages in both cases.
\end{example}
\begin{figure}[htbp]
 \centering
 \includegraphics[width=0.78\linewidth]{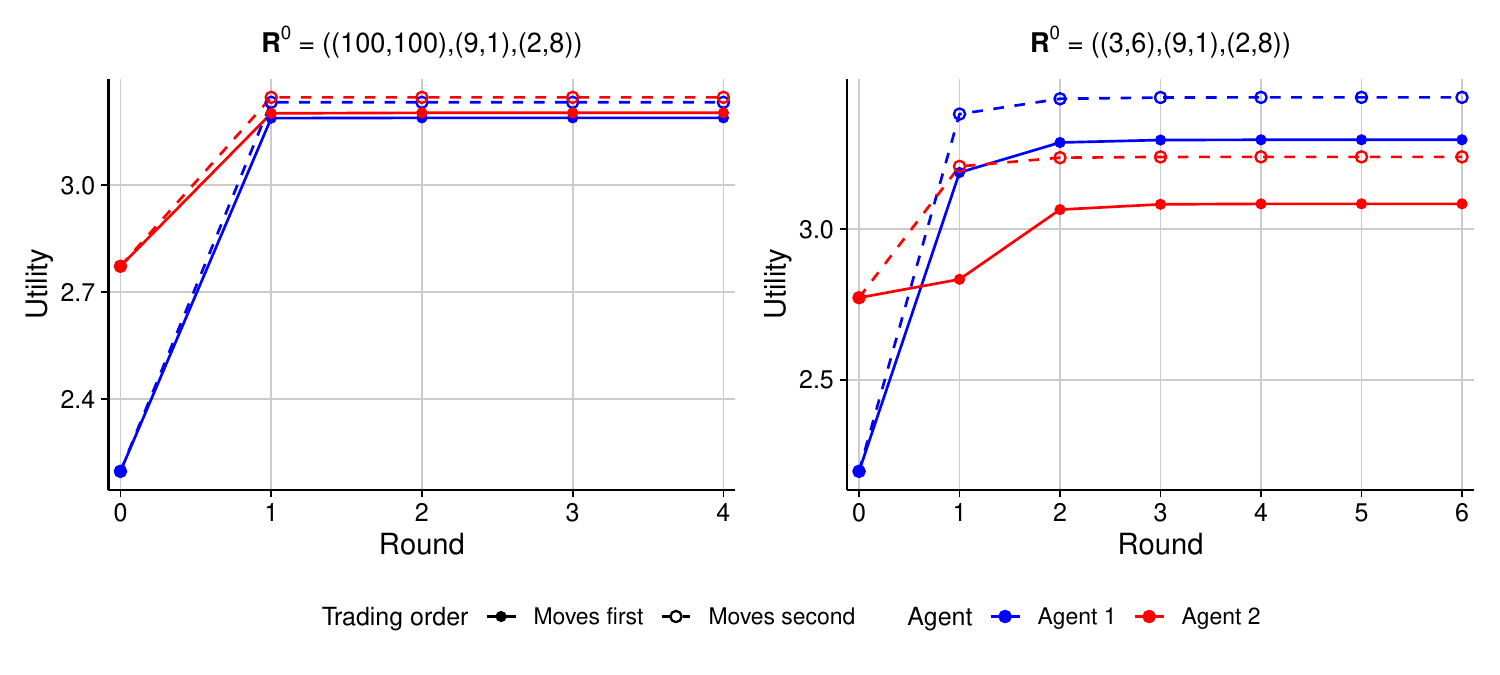}
 \caption{Utilities under alternating utility-maximizing trades.}
 \label{fig:util-alt_2}
\end{figure}

The next example examines how the limiting unilateral equilibrium compares with the Walrasian benchmark as the CFMM reserve scale varies.

\begin{example}\label{ex:liquidity-welfare}
Let $\varphi=u_1=u_2=\log x+\log y$, and fix the agents' initial holdings at $\mathbf r_1^0=(9,1)$ and $\mathbf r_2^0=(1,9)$. We vary the initial CFMM inventory as $\mathbf r_0^0=(\lambda,2\lambda)$, so the initial CFMM price remains $2$ across all calibrations. Agent~1 moves first, after which the agents alternate utility-maximizing trades. The corresponding Walrasian equilibrium assigns $(5,5)$ to each agent, with logarithmic utility $\log 25\approx3.219$.
Table~\ref{tab:market_states_2} reports late-stage iterates approximating the limiting unilateral equilibria. As the reserve scale $\lambda$ increases, the limiting CFMM price moves closer to its initial value of $2$. In these calibrations, agent~1's utility increases with $\lambda$, while agent~2's utility decreases. At $\lambda=1$, both agents attain approximately $3.234$, exceeding the Walrasian benchmark. At the other two reserve scales, only one agent exceeds that benchmark.
These welfare comparisons differ from the first-mover effect in Proposition~\ref{prop:first-mover-effect}, which compares trading orders for a fixed initial state. Although agent~1 faces a first-round disadvantage when moving first, her limiting utility can still exceed the Walrasian benchmark. This does not imply that the CFMM creates resources. The agents' aggregate holdings may change as they trade with the CFMM, while total system inventory and the CFMM invariant remain conserved.
\end{example}
\begin{table}[htbp]
\centering

\begin{subtable}{\linewidth}
\centering
\footnotesize
\setlength{\tabcolsep}{3pt}
\begin{tabular}{ccccccc}
\toprule
{} & $\mathbf r_0^t$ & $p_0^t$ & $\mathbf r_1^t$ & $u_1(\mathbf r_1^t)$ & $\mathbf r_2^t$ & $u_2(\mathbf r_2^t)$ \\
\midrule
$t=0$ & $(0.500,\,1.000)$ & 2.000 & $(9.000,\,1.000)$ & 2.197 & $(1.000,\,9.000)$ & 2.197 \\
$t=114$ & $(0.691,\,0.724)$ & 1.048 & $(4.855,\,5.086)$ & 3.207 & $(4.954,\,5.190)$ & 3.247 \\
\bottomrule
\end{tabular}
\caption{$\lambda=0.500$}
\end{subtable}

\vspace{1em}

\begin{subtable}{\linewidth}
\centering
\footnotesize
\setlength{\tabcolsep}{3pt}
\begin{tabular}{ccccccc}
\toprule
{} & $\mathbf r_0^t$ & $p_0^t$ & $\mathbf r_1^t$ & $u_1(\mathbf r_1^t)$ & $\mathbf r_2^t$ & $u_2(\mathbf r_2^t)$ \\
\midrule
$t=0$ & $(1.000,\,2.000)$ & 2.000 & $(9.000,\,1.000)$ & 2.197 & $(1.000,\,9.000)$ & 2.197 \\
$t=62$ & $(1.354,\,1.477)$ & 1.091 & $(4.823,\,5.262)$ & 3.234 & $(4.823,\,5.261)$ & 3.234 \\
\bottomrule
\end{tabular}
\caption{$\lambda=1.000$}
\end{subtable}

\vspace{1em}

\begin{subtable}{\linewidth}
\centering
\footnotesize
\setlength{\tabcolsep}{3pt}
\begin{tabular}{ccccccc}
\toprule
{} & $\mathbf r_0^t$ & $p_0^t$ & $\mathbf r_1^t$ & $u_1(\mathbf r_1^t)$ & $\mathbf r_2^t$ & $u_2(\mathbf r_2^t)$ \\
\midrule
$t=0$ & $(2.000,\,4.000)$ & 2.000 & $(9.000,\,1.000)$ & 2.197 & $(1.000,\,9.000)$ & 2.197 \\
$t=36$ & $(2.619,\,3.055)$ & 1.167 & $(4.815,\,5.617)$ & 3.297 & $(4.567,\,5.328)$ & 3.192 \\
\bottomrule
\end{tabular}
\caption{$\lambda=2.000$}
\end{subtable}

\caption{Late-stage market states and logarithmic utilities at different CFMM
reserve scales with $\mathbf r_0^{\mathrm{in}}=(\lambda,2\lambda)$.}
\label{tab:market_states_2}
\end{table}

Taken together, the numerical examples illustrate phenomena that are not captured by the static price-equality condition alone. Trading order, reserve scale, and the initial CFMM price can all affect the distribution of trader welfare. These observations are illustrative, and a general comparative-static theory for reserve scale is left for future work.

\section{Conclusion}\label{sec:con}
This paper develops an equilibrium framework for a closed CFMM economy.
Unlike the Walrasian benchmark, agents do not face a fixed linear budget
constraint. A trade moves the market maker's inventory along a liquidity
curve and changes the terms available to the next trader. We interpret a
market state as a unilateral no-trade equilibrium when it admits no profitable
unilateral valid deviation. Under our assumptions, an interior state is a
unilateral equilibrium exactly when the CFMM marginal price equals the
agents' marginal rates of substitution. For an interior initial state,
every individually rational unilateral equilibrium is Pareto optimal relative
to the fixed CFMM invariant.

The inventory mechanism also changes aggregation. A weak representative
agent can be constructed after an equilibrium is fixed through a weighted
sup-convolution of the agents' utilities. 
A state-independent strong representative agent exists precisely when aggregate frontiers collapse to
the level-curve family of a single utility, or equivalently when the agents
share a common homothetic preference.
Turning to trading dynamics, every interior feasible state is reachable by finitely many valid trades, while alternating utility-maximizing trades converge to a Pareto optimal unilateral equilibrium. Trading order can also generate first-mover advantages or disadvantages.

The model deliberately excludes fees, outside trading venues, strategic
liquidity provision, stochastic arrivals, and more than two assets. These
features are natural directions for future work. The central message of the
benchmark is that a CFMM does not simply reproduce a competitive market at a
different price. It replaces the fixed budget line by an endogenous
inventory curve. Marginal-price equality survives as a local equilibrium
condition, whereas welfare, aggregation, reachability, and path dependence
are governed by the shape and current position of that curve.

\bibliographystyle{apalike} 
\bibliography{ref.v20260921} 

\newpage

\appendix
\numberwithin{lemma}{section}
\numberwithin{equation}{section}

\section{Proofs for Section~\ref{sec:uf}}\label{app:technical-proofs}

 \subsection{Proof of Proposition~\ref{prop:level-curve}}
\label{app:utility-proofs}
Let $x'>0$. Since $u(x',0)=0$ and $\lim_{y\to\infty}u(x',y)=\infty$, continuity and the intermediate value theorem give a $y'$ such that $u(x',y')=c$. Strict monotonicity makes $y'$ unique. Define $\ell(x')=y'$. Then $s_{u,c}=\ell^g$, the graph of $\ell$.

\proofpart{Part~\eqref{curve:cts-dec}} Suppose $x_1<x_2$. If $\ell(x_1)\leq \ell(x_2)$, strict monotonicity of $u$ would give $u(x_1,\ell(x_1))<u(x_2,\ell(x_2))$, a contradiction. Thus $\ell$ is strictly decreasing. If $\ell$ were discontinuous at some $x\in\R_{++}$, its one-sided limits $y^-:=\ell(x^-)$ and $y^+:=\ell(x^+)$ would exist. Without loss of generality, suppose $y^-\neq y:=\ell(x)$. Continuity of $u$ would then imply $u(x,y^-)=c=u(x,y)$, again contradicting strict monotonicity.

\proofpart{Part~\eqref{curve:asymptote}} Let $x_n\downarrow0$. By part~\eqref{curve:cts-dec}, the sequence $\ell(x_n)$ is increasing. If $\ell(x_n)\uparrow M<\infty$, then $c=\lim_{n\to\infty}u(x_n,\ell(x_n))\leq\lim_{n\to\infty}u(x_n,M)=u(0,M)=0$, a contradiction. Similarly, if $x_n\uparrow\infty$ and $\ell(x_n)\downarrow m>0$, then $c=\lim_{n\to\infty}u(x_n,\ell(x_n))\geq\lim_{n\to\infty}u(x_n,m)=\infty$, which is impossible.

\proofpart{Part~\eqref{curve:ray-intercept}} Let $f(x)=\ell(x)/x$. The function $f$ is continuous, with $f(x)\to\infty$ as $x\downarrow0$ and $f(x)\to0$ as $x\to\infty$. Moreover, $f$ is strictly decreasing: if $0<x_1<x_2$, then $\ell(x_1)>\ell(x_2)>0$, and hence
\[
\frac{\ell(x_1)}{x_1}
>
\frac{\ell(x_2)}{x_1}
>
\frac{\ell(x_2)}{x_2}.
\]
The intermediate value theorem gives existence, and strict monotonicity gives uniqueness. Finally, if $k_1>k_2$ and $f(x_i)=k_i$, then $x_1<x_2$; otherwise $x_1\ge x_2$ would imply $k_1=f(x_1)\le f(x_2)=k_2$, a contradiction.

\proofpart{Part~\eqref{curve:cx}} Let $\lambda\in(0,1)$. Strict concavity gives $u((1-\lambda)x_1+\lambda x_2,(1-\lambda)\ell(x_1)+\lambda\ell(x_2))>c$. Strict monotonicity of $u$ then yields $\ell((1-\lambda)x_1+\lambda x_2)<(1-\lambda)\ell(x_1)+\lambda\ell(x_2)$, so $\ell$ is strictly convex.

\proofpart{Parts~\eqref{curve:diff} and~\eqref{curve:price}} Fix $x_0>0$ and let $y_0=\ell(x_0)$. Since $u$ is continuously differentiable on $\R_{++}^2$ and $\partial u/\partial y(x_0,y_0)>0$, the implicit function theorem makes $\ell$ continuously differentiable near $x_0$. Because $x_0$ was arbitrary, $\ell$ is continuously differentiable on $\R_{++}$. Moreover, $\ell'(x)=-(\partial u/\partial x)/(\partial u/\partial y)(x,\ell(x))=-p(u,x,\ell(x))$.

\proofpart{Part~\eqref{curve:axis-tangent}} Since $\ell$ is strictly convex, $\ell'$ is strictly increasing, so $\lim_{x\downarrow0}\ell'(x)$ exists. If this limit were $M>-\infty$, then $\ell(0^+)=\ell(1)-\int_0^1\ell'(x)\dd x\leq\ell(1)-M$, where the integral is improper. This contradicts part~\eqref{curve:asymptote}. Likewise, $\ell(x)\to0$ as $x\to\infty$ implies $\ell'(x)\to0$. Continuity and strict monotonicity of $\ell'$ now show that, for each $p>0$, there is a unique $x>0$ such that $\ell'(x)=-p$.

\subsection{Proof of Proposition~\ref{prop:equivalent-utility}}
By Assumption~\eqref{as:1}, both utilities vanish on the coordinate axes
and are strictly positive in the interior, so comparisons involving
the zero level are preserved. It therefore suffices to consider
$r^\alpha,r^\beta\in\mathbb{R}_{++}^2$.

First suppose that $u(\mathbf r^\alpha)=u(\mathbf r^\beta)$. The two points lie on the
same level curve of $u$. Since $\mathcal L_u=\mathcal L_v$, that curve is
also a level curve of $v$, and therefore $v(\mathbf r^\alpha)=v(\mathbf r^\beta)$.

Next suppose that $u(\mathbf r^\alpha)<u(\mathbf r^\beta)$. Set
$c=u(\mathbf r^\alpha)$, $c'=u(\mathbf r^\beta)$, $d=v(\mathbf r^\alpha)$, and
$d'=v(\mathbf r^\beta)$. Equality of the level-curve families gives
$s_{u,c}=s_{v,d}$ and $s_{u,c'}=s_{v,d'}$. Since $u$ is strictly
increasing and $c<c'$, the curve $s_{u,c'}$ lies strictly above
$s_{u,c}$. Strict monotonicity of $v$ therefore gives $d<d'$. The case
$u(\mathbf r^\alpha)>u(\mathbf r^\beta)$ is analogous.

\subsection{Proof of Proposition~\ref{prop:homothetic-preference}}
\proofpart{(i) \(\Rightarrow\) (ii)} Let $f$ satisfy $u=f\circ v$, put $d=f^{-1}(c)$ and $d'=f^{-1}(c')$, and take $k=d/d'$. 
Then
\[
\begin{aligned}
(x,y)\in\ell^g_{u,c'}
&\iff (x,y)\in\ell^g_{v,d'}\\
&\iff (kx,ky)\in\ell^g_{v,d}
\iff (kx,ky)\in\ell^g_{u,c}.
\end{aligned}
\]

\proofpart{(ii) \(\Rightarrow\) (i)} Fix $c'=1$. For any $(x,y)\in\R_{++}^2$, Proposition~\ref{prop:level-curve}~\eqref{curve:ray-intercept} gives a unique $x'$ such that $\ell_1(x')/x'=y/x$. Define $v(x,y)=x/x'$ and set $v(x,0)=v(0,y)=v(0,0)=0$.

For each $c>0$, there is a scalar $k(c)>0$ such that $\ell_{u,c}^g=k(c)\ell_{u,1}^g$. This scalar is unique because each positive ray meets each level curve exactly once. If $c_1<c_2$, strict monotonicity of $u$ implies $k(c_1)<k(c_2)$. If $u(x,y)=c$, then $(x,y)=k(c)(x',\ell_{u,1}(x'))$ for the ray-intersection point used above, so $v(x,y)=k(c)$. Thus $v=k\circ u$ on $\R_{++}^2$, with $k$ strictly increasing, and the two utilities share their lowest level on the coordinate axes. Hence $v$ is equivalent to $u$. The definition also gives $v(ax,ay)=av(x,y)$ for every $a>0$, so $v$ is homogeneous.

\proofpart{(ii) \(\Rightarrow\) (iii)} Fix $x,y,k>0$, and let $c'=u(x,y)$ and $c=u(kx,ky)$. Choose $k'>0$ such that $\ell_c^g=k'\ell_{c'}^g$. Then $k'y=k'\ell_{c'}(x)=\ell_c(k'x)$. Proposition~\ref{prop:level-curve}~\eqref{curve:ray-intercept} gives $k'=k$, and differentiation yields $-p(u,x,y)=\ell'_{c'}(x)=k^{-1}k\ell'_c(kx)=-p(u,kx,ky)$.

\proofpart{(iii) \(\Rightarrow\) (ii)} Fix $c,c'>0$ and $(x,y)\in\R_{++}^2$ with $u(x,y)=c$. Proposition~\ref{prop:level-curve}~\eqref{curve:ray-intercept} gives $(x',y')\in\R_{++}^2$ with $u(x',y')=c'$ and $y'/x'=y/x$. Put $k=x/x'>0$ and define $f(z)=k\ell_{c'}(z/k)$. Then $f(x)=y$, and for every $z>0$,
\[
f'(z)=\ell'_{c'}\left(\frac{z}{k}\right)=-p\left(u,\frac{z}{k},\ell_{c'}\left(\frac{z}{k}\right)\right)=-p(u,z,f(z)),
\]
where the final equality uses (iii). Consequently,
\[
\frac{d}{dz}u(z,f(z))
=\frac{\partial u}{\partial y}(z,f(z))\bigl(p(u,z,f(z))+f'(z)\bigr)=0.
\]
Thus $u(z,f(z))$ is constant. Since $f(x)=y$ and $u(x,y)=c$, we have $u(z,f(z))=c$ for all $z>0$. Uniqueness of the level-curve representation gives $f(z)=\ell_c(z)$, hence $\ell_c(z)=k\ell_{c'}(z/k)$ and $\ell_c^g=k\ell_{c'}^g$.

\proofpart{Part~\eqref{homo:1}} Suppose $0<x_1\leq x_2$ and $y_1\geq y_2>0$, with at least one inequality strict. Then $y_1/x_1>y_2/x_2$. Let $c=u(x_2,y_2)$, and choose $x\in\R_{++}$ such that $\ell_c(x)/x=y_1/x_1$. Proposition~\ref{prop:level-curve}~\eqref{curve:ray-intercept} implies $x<x_2$. Therefore $p(u,x_1,y_1)=p(u,x,\ell_c(x))=-\ell'_c(x)>-\ell'_c(x_2)=p(u,x_2,y_2)$.

\proofpart{Part~\eqref{homo:2}} Define $g(r)=p(u,1,r)$ for $r>0$. By (iii), $p(u,x,y)=p(u,1,y/x)=g(y/x)$. Part~\eqref{homo:1} shows that $r_1<r_2$ implies $g(r_1)<g(r_2)$, so $g$ is strictly increasing.


\section{Proofs for Section~\ref{sec:eq0}}

\subsection{Proof of Proposition \ref{prop:single-trade}}
Let $c=\varphi(\mathbf r_0)$ be the CFMM invariant level and let $\mathbf T=\mathbf r_0+\mathbf r=(T_x,T_y)$ be total inventory in the single-agent economy. It is convenient to parameterize a trade $\mathbf h$ by the post-trade CFMM inventory $\mathbf z=\mathbf r_0+\mathbf h$.

\proofpart{Part~\eqref{trade:exist}} The feasible post-trade CFMM inventories form $K:=\{\mathbf z\in\R_+^2:\varphi(\mathbf z)=c,\ 0\leq\mathbf z\leq\mathbf T\}$, where inequalities are coordinatewise. This set is nonempty because it contains $\mathbf r_0$, and it is compact as a closed subset of $[0,T_x]\times[0,T_y]$. Hence the continuous function $\mathbf z\mapsto u(\mathbf T-\mathbf z)$ attains a maximum on $K$.

Every maximizer is interior. The zero trade is valid, so the positive value $u(\mathbf r)$ is attainable. If $\mathbf T-\mathbf z$ lay on the boundary of $\R_+^2$, its utility would be zero by Assumption~\eqref{as:1}. Also, $\mathbf z$ cannot lie on the boundary because $\varphi(\mathbf z)=c>0$, whereas $\varphi$ is zero on the coordinate axes.

\proofpart{Part~\eqref{trade:optimal-implies-price}} Write the post-trade CFMM inventory as $\mathbf z=(x,\ell_{\varphi,c}(x))$. The corresponding agent holding is $\mathbf T-\mathbf z=(T_x-x,T_y-\ell_{\varphi,c}(x))$, so the problem reduces to maximizing $F(x):=u(T_x-x,T_y-\ell_{\varphi,c}(x))$ on the relevant interval. At an interior maximizer $x^*$, the chain rule gives
\[
F'(x)=-\frac{\partial u}{\partial x}(T_x-x,T_y-\ell_{\varphi,c}(x))
-\frac{\partial u}{\partial y}(T_x-x,T_y-\ell_{\varphi,c}(x))\ell'_{\varphi,c}(x).
\]
Using $\ell'_{\varphi,c}(x)=-p(\varphi,x,\ell_{\varphi,c}(x))$, this becomes
\[
F'(x)=\frac{\partial u}{\partial y}(T_x-x,T_y-\ell_{\varphi,c}(x))
\left[p(\varphi,x,\ell_{\varphi,c}(x))-p(u,T_x-x,T_y-\ell_{\varphi,c}(x))\right].
\]
Since $\partial u/\partial y>0$, the condition $F'(x^*)=0$ implies $p(\varphi,x^*,\ell_{\varphi,c}(x^*))=p(u,T_x-x^*,T_y-\ell_{\varphi,c}(x^*))$. Equivalently, $p(\varphi,\mathbf r_0+\mathbf h^*)=p(u,\mathbf r-\mathbf h^*)$.

\proofpart{Parts~\eqref{trade:unique}--\eqref{trade:sign}} As $x$ increases along the liquidity curve, the CFMM holds more $X$ and less $Y$, so $p(\varphi,x,\ell_{\varphi,c}(x))$ is strictly decreasing. The agent then holds less $X$ and more $Y$, so Assumption~\eqref{as:4} makes $p(u,T_x-x,T_y-\ell_{\varphi,c}(x))$ strictly increasing. Hence $D(x):=p(\varphi,x,\ell_{\varphi,c}(x))-p(u,T_x-x,T_y-\ell_{\varphi,c}(x))$ is strictly decreasing. The equation $D(x)=0$ has at most one solution. Every utility maximizer is interior and satisfies this equation, so the utility-maximizing trade is unique; conversely, any interior valid trade satisfying the price equality is that unique maximizer.

The zero trade corresponds to $x=x_0$ when $\mathbf r_0=(x_0,y_0)$. It is the unique maximizer exactly when $D(x_0)=p(\varphi,\mathbf r_0)-p(u,\mathbf r)=0$. If $p(\varphi,\mathbf r_0)>p(u,\mathbf r)$, then $D(x_0)>0$, and strict decrease of $D$ implies $x^*>x_0$. Thus $h_x^*=x^*-x_0>0$ and $h_y^*=\ell_{\varphi,c}(x^*)-\ell_{\varphi,c}(x_0)<0$, so $\operatorname{sgn}(\mathbf h^*)=(1,-1)$. The opposite price inequality gives $\operatorname{sgn}(\mathbf h^*)=(-1,1)$.

\subsection{Proof of Theorem \ref{thm:1}}
If $p(\varphi,\mathbf r_0^\alpha)=p(u_1,\mathbf r_1^\alpha)=p(u_2,\mathbf r_2^\alpha)$,
Proposition~\ref{prop:single-trade} says that both agents' unique
utility-maximizing trades are the zero trade, so neither agent has a strictly
improving valid trade. Conversely, if no valid trade strictly improves
either agent's utility, each agent's utility-maximizing trade must be
zero. Proposition~\ref{prop:single-trade} then gives the same price
equality.

\subsection{Proof of Proposition \ref{prop:unique-price}}
Suppose, without loss of generality, that $p_0^\alpha>p_0^\beta$. Proposition~\ref{prop:homothetic-preference}~\eqref{homo:2} gives a common strictly increasing ratio-price function $g$ satisfying $g(y/x)=p(u_1,x,y)=p(u_2,x,y)$. Therefore
\[
\frac{y_A^\alpha}{x_A^\alpha}=\frac{y_1^\alpha}{x_1^\alpha}=\frac{y_2^\alpha}{x_2^\alpha}
>\frac{y_1^\beta}{x_1^\beta}=\frac{y_2^\beta}{x_2^\beta}=\frac{y_A^\beta}{x_A^\beta}.
\]
On the other hand, $x_0^\alpha<x_0^\beta$ and $y_0^\alpha>y_0^\beta$. Conservation then gives $x_A^\alpha>x_A^\beta$ and $y_A^\alpha<y_A^\beta$, contradicting the displayed ratio inequality.

\subsection{Proof of Lemma~\ref{lem:compact-feasible-improvement}}
\label{app:compact-feasible-improvement}

Locally write $\mathbf r_T:=\mathbf r_T^{\mathrm{in}}$ for the common
total inventory. Every feasible state belongs to the compact box
\[
[0,\mathbf r_T]^3
:=\left\{\mathbf R^\alpha\in(\R_+^2)^3:
0\leq\mathbf r_i^\alpha\leq\mathbf r_T,\ i=0,1,2\right\}.
\]
The feasible set is closed in this box because it is defined by
$
\mathbf r_T^\alpha=\mathbf r_T$,  $
\varphi(\mathbf r_0^\alpha)=\varphi(\mathbf r_0^{\mathrm{in}}),
$
and $\varphi$ is continuous. Hence
$\mathcal F(\mathbf R^{\mathrm{in}})$ is compact. Continuity of $u_1,u_2$
makes
\[
\mathcal U(\mathbf R^{\mathrm{in}})
=\left\{\mathbf R^\alpha\in\mathcal F(\mathbf R^{\mathrm{in}}):
u_i(\mathbf r_i^\alpha)\geq u_i(\mathbf r_i^{\mathrm{in}}),\ i=1,2\right\}
\]
a closed subset of $\mathcal F(\mathbf R^{\mathrm{in}})$. It is nonempty
because it contains $\mathbf R^{\mathrm{in}}$.

\subsection{Proof of Proposition \ref{prop:2}}
\proofpart{Part~(i)} By Lemma~\ref{lem:compact-feasible-improvement},
$\mathcal U(\mathbf R^{\mathrm{in}})$ is nonempty and compact. Choose the collection of states  that maximizes agent~1's utility on this set, and then,
among all such maximizers, choose $\mathbf R^*$ to maximize agent~2's
utility. If another state in $\mathcal U(\mathbf R^{\mathrm{in}})$ weakly improved both agents and strictly improved
one, maximality of agent~1's utility would force equality for agent~1, while
the second-stage choice would rule out strict improvement for agent~2.
Thus $\mathbf R^*$ is undominated within $\mathcal U(\mathbf R^{\mathrm{in}})$.
Since $\mathbf R^*\in\mathcal U(\mathbf R^{\mathrm{in}})$, any state in
$\mathcal F(\mathbf R^*)=\mathcal F(\mathbf R^{\mathrm{in}})$ that weakly
improves both agents also belongs to $\mathcal U(\mathbf R^{\mathrm{in}})$.
Therefore $\mathbf R^*$ is Pareto optimal.

\proofpart{Part~(ii)} If a Pareto optimal state were not a unilateral equilibrium,
one agent would have a valid trade that strictly raises her utility while
leaving the other agent's holding unchanged. The resulting state would
remain in $\mathcal F(\mathbf R^{\mathrm{in}})$ and would strictly Pareto
dominate the original state, a contradiction.

\subsection{Proof of Theorem \ref{thm:unilateral-implies-Pareto}}
Let $\mathbf R^*\in\mathcal U(\mathbf R^{\mathrm{in}})$ be an arbitrary unilateral equilibrium. It is interior because the initial state is
interior.
Write $u_0=\varphi$, let $p_0^*$ denote the common price given by
Theorem~\ref{thm:1}, and put
\[
\mathbf q=(p_0^*,1),\quad
a_i=\partial_y u_i(\mathbf r_i^*)>0,\quad i=0,1,2.
\]
Then $\nabla u_i(\mathbf r_i^*)=a_i\mathbf q$ for $i\in\{0,1,2\}$.

Suppose that $\mathbf R^*$ is not Pareto optimal. By definition, there exists
$\mathbf R^\alpha\in\mathcal F(\mathbf R^*)$
that is a Pareto improvement over $\mathbf R^*$.
The market and the agents' positive utility levels imply that
$\mathbf R^\alpha$ is also interior.
By concavity, for $i\in\{0,1,2\}$,
\[
u_i(\mathbf r_i^\alpha)-u_i(\mathbf r_i^*)
\leq
\nabla u_i(\mathbf r_i^*)\cdot
(\mathbf r_i^\alpha-\mathbf r_i^*).
\]
Hence,
\[
\frac{u_i(\mathbf r_i^\alpha)-u_i(\mathbf r_i^*)}{a_i}
\leq
\mathbf q\cdot(\mathbf r_i^\alpha-\mathbf r_i^*),
\qquad i=0,1,2.
\]
The left-hand side is zero for the CFMM, nonnegative for both agents, and
strictly positive for at least one agent. Summing and using feasibility yields
\begin{equation*}
0<\sum_{i=0}^2\frac{u_i(\mathbf r_i^\alpha)-u_i(\mathbf r_i^*)}{a_i}
\leq\mathbf q\cdot(\mathbf r_T^\alpha-\mathbf r_T^*)=0,
\end{equation*}
a contradiction. Therefore $\mathbf R^*$ is Pareto optimal.

\section{Proofs for Section \ref{sec:cba}}

\subsection{Proof of Lemma~\ref{lem:infimal-convolution-regularity}}
\label{app:infimal-convolution-regularity}
Write \(f=\ell_{u_1,c_1}\) and \(g=\ell_{u_2,c_2}\). For fixed \(x>0\),
consider \(H_x(a)=f(a)+g(x-a)\) on \(a\in(0,x)\). By
Proposition~\ref{prop:level-curve}, \(H_x(a)\to\infty\) as either
\(a\downarrow0\) or \(a\uparrow x\). Hence the minimum is attained in the
interior. Strict convexity of \(f\) and \(g\) makes the minimizer unique;
write it as \(a(x)\), and put \(x_1(x)=a(x)\) and
\(x_2(x)=x-a(x)\). The first-order condition is
\(f'(x_1(x))=g'(x_2(x))\).

The minimizer depends continuously on \(x\). Indeed, if \(x_k\to x\), the values at the fixed feasible split $x_k/2$ remain bounded. A minimizing sequence can therefore approach neither endpoint, because $f(a)\to\infty$ as $a\downarrow0$ and $g(x_k-a)\to\infty$ as $a\uparrow x_k$. Every convergent subsequence of \(a(x_k)\) must then minimize \(H_x\); uniqueness forces its limit to be \(a(x)\). Thus the whole sequence converges.

Let \(F=f\boxplus g\). For small \(h\), optimality at \(x\) and \(x+h\)
gives
\[
F(x+h)-F(x)
\leq
g(x_2(x)+h)-g(x_2(x))
\]
and
\[
F(x+h)-F(x)
\geq
g(x+h-a(x+h))-g(x-a(x+h)).
\]
Taking the right- and left-hand difference quotients separately, and using continuity of \(a(\cdot)\) and \(g'\), shows that \(F'(x)=g'(x_2(x))\). The first-order condition also gives
\(F'(x)=f'(x_1(x))\), and continuity of the minimizing split makes \(F'\)
continuous.

To prove strict convexity, take \(x\neq y\), their unique minimizing splits,
and \(\theta\in(0,1)\). The convex combination of the two splits is feasible
for \(\theta x+(1-\theta)y\). At least one component of the two splits
differs, so strict convexity of \(f\) and \(g\) yields
\[
F(\theta x+(1-\theta)y)
<
\theta F(x)+(1-\theta)F(y).
\]
The price identity follows from
Proposition~\ref{prop:level-curve}~\eqref{curve:price}.

\subsection{Proof of Lemma~\ref{lem:weighted-sup-convolution}}
\label{app:weighted-sup-convolution}
For a fixed \(\mathbf r\in\R_+^2\), the allocation set
\[
\mathcal A(\mathbf r)
=
\left\{
(\mathbf r_1,\ldots,\mathbf r_n)\in(\R_+^2)^n:
\sum_{i=1}^n\mathbf r_i=\mathbf r
\right\}
\]
is nonempty and compact. The objective is continuous, so the supremum is
finite and attained.

To prove concavity, let \(\mathbf r,\mathbf s\in\R_+^2\), choose optimal
allocations \((\mathbf r_i)\) and \((\mathbf s_i)\), and fix
\(\theta\in[0,1]\). The allocation
\((\theta\mathbf r_i+(1-\theta)\mathbf s_i)_{i=1}^n\) is feasible for
\(\theta\mathbf r+(1-\theta)\mathbf s\). Concavity of the \(u_i\) gives
\[
U^w(\theta\mathbf r+(1-\theta)\mathbf s)
\geq
\theta U^w(\mathbf r)+(1-\theta)U^w(\mathbf s).
\]

For continuity, let \(\mathbf r^k\to\mathbf r\), and choose a subsequence along which \(U^w(\mathbf r^k)\) converges to its limsup. If \((\mathbf r_i^k)\) is optimal for \(\mathbf r^k\), uniform boundedness gives a further subsequence converging to an allocation \((\overline{\mathbf r}_i)\) of \(\mathbf r\). Continuity of the objective yields
\(\limsup_k U^w(\mathbf r^k)\leq U^w(\mathbf r)\). Conversely, take an
optimal allocation \((\mathbf r_i)\) of \(\mathbf r\). Coordinate by
coordinate, scale the positive aggregate coordinates proportionally from
\(r_j\) to \(r_j^k\); when \(r_j=0\), allocate the vanishing quantity
\(r_j^k\) to one agent. This produces feasible allocations
\((\widetilde{\mathbf r}_i^k)\) of \(\mathbf r^k\) with
\(\widetilde{\mathbf r}_i^k\to\mathbf r_i\). Continuity of the utilities
then gives
\(\liminf_k U^w(\mathbf r^k)\geq U^w(\mathbf r)\).

Finally, concavity and differentiability at the interior points
\(\mathbf r_i^*\) imply, for every feasible split
\(\sum_i\mathbf r_i=\mathbf r\),
\[
\sum_{i=1}^n w_i u_i(\mathbf r_i)
\leq
\sum_{i=1}^n w_i u_i(\mathbf r_i^*)
+
\sum_{i=1}^n
w_i\nabla u_i(\mathbf r_i^*)\cdot(\mathbf r_i-\mathbf r_i^*).
\]
Under the common-gradient condition, the last sum equals
\(\boldsymbol\lambda\cdot(\mathbf r-\mathbf r_A^*)\). At
\(\mathbf r=\mathbf r_A^*\), this proves optimality of the proposed
allocation. Taking the supremum over all splits for an arbitrary
\(\mathbf r\) gives~\eqref{eq:weighted-sup-supergradient}.

\subsection{Proof of Proposition \ref{prop:weak-rep}}
Theorem~\ref{thm:1} gives
\(p(\varphi,\mathbf r_0^*)=p(u_i,\mathbf r_i^*)=p_0^*\) for every
\(i\in[n]\). Hence the choice
\(w_i=1/\partial_y u_i(\mathbf r_i^*)\) yields
\(w_i\nabla u_i(\mathbf r_i^*)=(p_0^*,1)\). Applying
Lemma~\ref{lem:weighted-sup-convolution} with
\(\boldsymbol\lambda=(p_0^*,1)\) shows that the equilibrium allocation
attains the supremum in~\eqref{eq:sup_convolution_rep} at
\(\mathbf r_A^*=\sum_i\mathbf r_i^*\), and that
\[
u_A^w(\mathbf r_A)
\leq
u_A^w(\mathbf r_A^*)+(p_0^*,1)\cdot(\mathbf r_A-\mathbf r_A^*).
\]

By feasibility, the aggregate trade
\(\mathbf h^*=\mathbf r_0^*-\mathbf r_0^{\mathrm{in}}\) moves the
representative agent from \(\mathbf r_A^{\mathrm{in}}\) to
\(\mathbf r_A^*\). In representative-agent coordinates, the feasible
frontier is strictly concave and has slope \(-p_0^*\) at
\(\mathbf r_A^*\). Therefore every attainable aggregate holding
\(\mathbf r_A\neq\mathbf r_A^*\) satisfies $
(p_0^*,1)\cdot(\mathbf r_A-\mathbf r_A^*)<0$.
The supergradient inequality gives
\(u_A^w(\mathbf r_A)<u_A^w(\mathbf r_A^*)\). Since distinct trades produce
distinct aggregate holdings, \(\mathbf h^*\) uniquely solves the
representative agent's problem.

\subsection{Proof of Lemma~\ref{lem:cont-of-transform}}
\label{app:tangent-transform}
Fix $q>0$. Proposition~\ref{prop:level-curve} gives $\ell_{u,c}(x)\to\infty$ as $x\downarrow0$ and $\ell_{u,c}(x)\to0$ as $x\to\infty$. Hence $x\mapsto qx+\ell_{u,c}(x)$ tends to infinity at both endpoints. The infimum defining $A_{\ell_{u,c}}(q)$ is therefore attained and lies in $(0,\infty)$.

We first prove strict monotonicity. Let $0<c_1<c_2$. For every $x>0$, strict monotonicity of $u$ gives $\ell_{u,c_1}(x)<\ell_{u,c_2}(x)$. If $x_2$ minimizes $A_{\ell_{u,c_2}}(q)$, then
\[
A_{\ell_{u,c_2}}(q)=qx_2+\ell_{u,c_2}(x_2)>qx_2+\ell_{u,c_1}(x_2)\geq A_{\ell_{u,c_1}}(q).
\]
Thus $c\mapsto A_{\ell_{u,c}}(q)$ is strictly increasing.

Next we prove continuity. Fix $c_0>0$ and choose $0<c_-<c_0<c_+$. For $c\in[c_-,c_+]$, let $x_c$ minimize $x\mapsto qx+\ell_{u,c}(x)$. Monotonicity gives $A_{\ell_{u,c}}(q)\leq A_{\ell_{u,c_+}}(q)=:M$, so $qx_c\leq M$, $\ell_{u,c}(x_c)\leq M$, and $x_c\leq M/q$. Since $\ell_{u,c}(x)\geq\ell_{u,c_-}(x)$ and $\ell_{u,c_-}(x)\to\infty$ as $x\downarrow0$, there is an $a>0$ such that no minimizer lies in $(0,a)$. Therefore every $x_c$ lies in the compact interval $K=[a,M/q]$.

The map $(c,x)\mapsto\ell_{u,c}(x)$ is continuous on $[c_-,c_+]\times K$. Indeed, it is defined implicitly by $u(x,\ell_{u,c}(x))=c$, and the implicit function theorem gives continuity at every point. Compactness then makes this continuity uniform on the rectangle. If $c_n\to c_0$, let $x_n$ and $x_0$ be corresponding minimizers. Then
\[
A_{\ell_{u,c_n}}(q)-A_{\ell_{u,c_0}}(q)
\leq\sup_{x\in K}\left|\ell_{u,c_n}(x)-\ell_{u,c_0}(x)\right|,
\]
and the same bound holds with the two $A$-terms reversed. Uniform convergence on $K$ therefore gives $A_{\ell_{u,c_n}}(q)\to A_{\ell_{u,c_0}}(q)$.

It remains to show that $\alpha_u(c)=A_{\ell_{u,c}}(q_0)$ maps onto $(0,\infty)$. First, $\alpha_u(c)\downarrow0$ as $c\downarrow0$. Given $\varepsilon>0$, choose $x,y>0$ with $q_0x+y<\varepsilon$ and set $c_\varepsilon=u(x,y)$. If $0<c\leq c_\varepsilon$, then $0<\alpha_u(c)\leq q_0x+\ell_{u,c}(x)\leq q_0x+y<\varepsilon$. Second, $\alpha_u(c)\to\infty$ as $c\to\infty$. Otherwise, there would be $c_n\to\infty$ and $M<\infty$ with $\alpha_u(c_n)\leq M$. If $x_n$ is a minimizer and $y_n=\ell_{u,c_n}(x_n)$, then $0\leq x_n\leq M/q_0$ and $0\leq y_n\leq M$. The points $(x_n,y_n)$ lie in a compact set, so continuity of $u$ makes $u(x_n,y_n)$ bounded, contradicting $u(x_n,y_n)=c_n\to\infty$.

The map $\alpha_u$ is continuous and strictly increasing, with endpoint limits $0$ and $\infty$, so it is onto $(0,\infty)$. Existence and uniqueness of the level curve with $A_{\ell_{u,c}}(q_0)=t$ follow immediately.

\subsection{Proof of Lemma~\ref{lem:collapse}}
\label{app:collapse-proof}
Suppose first that $u_1,u_2$, and $v$ represent the same homothetic preference. Then $\Lambda_{u_1}=\Lambda_{u_2}=\Lambda_v$. Let $\ell$ be one level curve in this common family. Proposition~\ref{prop:homothetic-preference} implies that arbitrary level curves have the form $f(x)=a\ell(x/a)$ and $g(x)=b\ell(x/b)$ for some $a,b>0$. Convexity of $\ell$ gives
\begin{align*}
(f\boxplus g)(x)
&=\inf_{x_1+x_2=x}\left\{a\ell(x_1/a)+b\ell(x_2/b)\right\}\\
&=(a+b)\ell\left(\frac{x}{a+b}\right),
\end{align*}
with equality at $x_1=ax/(a+b)$ and $x_2=bx/(a+b)$. Thus $f\boxplus g\in\Lambda_v$. Conversely, every positive level curve of $v$ is obtained by taking $a=b$ equal to half its dilation factor. Hence~\eqref{eq:collapse_statement} holds.

Conversely, assume~\eqref{eq:collapse_statement}. For each $q>0$, the infimum defining $A_f(q)$ is attained at the unique $x_q$ satisfying $f'(x_q)=-q$.

The supporting-line representation $f(x)=\sup_{q>0}\{A_f(q)-qx\}$ shows that $f\mapsto A_f$ is injective. The transform is also additive: for every $q>0$, $A_{f\boxplus g}(q)=A_f(q)+A_g(q)$. Indeed,
\begin{align*}
A_{f\boxplus g}(q)
&=\inf_x\{qx+(f\boxplus g)(x)\}\\
&=\inf_{x_1,x_2}\{q(x_1+x_2)+f(x_1)+g(x_2)\}\\
&=A_f(q)+A_g(q).
\end{align*}

Let $\mathcal A_u:=\{A_{\ell_{u,c}}:c>0\}$, with pointwise addition. The collapse condition is equivalent to $\mathcal A_{u_1}+\mathcal A_{u_2}=\mathcal A_v$. As $c\downarrow0$, $A_{\ell_{u,c}}(q)\downarrow0$ for every $q>0$.

We first prove $\mathcal A_{u_1}\subseteq\mathcal A_v$ and $\mathcal A_{u_2}\subseteq\mathcal A_v$. Fix $A\in\mathcal A_{u_1}$ and choose $C_n\in\mathcal A_{u_2}$ with $C_n(q)\downarrow0$ for every $q>0$. Then $A+C_n\in\mathcal A_v$. Their values at a fixed $q_0$ converge to $A(q_0)$; Lemma~\ref{lem:cont-of-transform} and the one-to-one parametrization by the $q_0$-intercept imply that $A+C_n$ converges pointwise to the unique member of $\mathcal A_v$ with that intercept. Since it also converges pointwise to $A$, we obtain $A\in\mathcal A_v$. The argument for $\mathcal A_{u_2}$ is identical.

The inclusions are equalities. For fixed $q_0>0$, each family $\mathcal A_u$ contains exactly one element for each prescribed positive value at $q_0$. Hence $\mathcal A_{u_1}=\mathcal A_{u_2}=\mathcal A_v$, and in particular $\mathcal A_v+\mathcal A_v=\mathcal A_v$.

For $t>0$, let $B_t\in\mathcal A_v$ be the unique element with $B_t(q_0)=t$. Closure under addition gives $B_s+B_t=B_{s+t}$ for $s,t>0$. For each fixed $q>0$, continuity and additivity imply $B_t(q)=tB_1(q)$. Let $f_t$ be the unique level curve with $A_{f_t}(q_0)=t$. Then $A_{f_t}=tA_{f_1}$. The dilated curve $g_t(x)=tf_1(x/t)$ has the same transform, so injectivity yields $f_t(x)=tf_1(x/t)$. Thus the level curves of $v$ are mutual dilations, and $v$ represents a homothetic preference.

Finally, $\mathcal A_{u_1}=\mathcal A_{u_2}=\mathcal A_v$ and injectivity imply $\Lambda_{u_1}=\Lambda_{u_2}=\Lambda_v$, and hence $\mathcal L_{u_1}^+=\mathcal L_{u_2}^+=\mathcal L_v^+$. The zero level sets also coincide with the coordinate axes. Proposition~\ref{prop:equivalent-utility} therefore shows that $u_1,u_2$, and $v$ are equivalent.

\subsection{Proof of Theorem \ref{thm:strong-rep}}

Lemma~\ref{lem:collapse} gives the equivalence of (ii) and (iii). We prove the equivalence between (i) and (ii).

\proofpart{(ii) \(\Rightarrow\) (i)}
Let \(\mathbf R^*\in\mathcal U(\mathbf R^{\mathrm{in}})\) be a unilateral
equilibrium, set \(c_i=u_i(\mathbf r_i^*)\), and write
\[
\mathbf r_A^*=\mathbf r_1^*+\mathbf r_2^*,
\qquad
F=\ell_{u_1,c_1}\boxplus\ell_{u_2,c_2}.
\]
Theorem~\ref{thm:1} gives a common price \(p^*\).
Lemma~\ref{lem:infimal-convolution-regularity} then shows that the decomposition
\(\mathbf r_A^*=\mathbf r_1^*+\mathbf r_2^*\) is the unique minimizing split defining \(F\), and
that \(F\) has slope \(-p^*\) at the $x$-coordinate of \(\mathbf r_A^*\). Under (ii), \(F\) is a
level curve of \(v\), so
\(p(v,\mathbf r_A^*)=p^*=p(\varphi,\mathbf r_0^*)\). The representative
agent's level curve is tangent to her feasible frontier at
\(\mathbf r_A^*\). Proposition~\ref{prop:single-trade} therefore makes the
aggregate trade uniquely utility-maximizing, and \(v\) is a strong representative
agent.

\proofpart{(i) \(\Rightarrow\) (ii)}
Fix \(c_1,c_2>0\), set
\(F=\ell_{u_1,c_1}\boxplus\ell_{u_2,c_2}\), and take
\(\mathbf z=(x,F(x))\) on its graph. By
Lemma~\ref{lem:infimal-convolution-regularity}, the unique minimizing split
\(\mathbf z=\mathbf r_1+\mathbf r_2\) satisfies
\(p(u_1,\mathbf r_1)=p(u_2,\mathbf r_2)=:p\). Choose a CFMM inventory
\(\mathbf r_0\) on a liquidity curve with
\(p(\varphi,\mathbf r_0)=p\). Then
\((\mathbf r_0,\mathbf r_1,\mathbf r_2)\) is a unilateral equilibrium.
Taking this state as the initial state, strong representation makes the zero
aggregate trade optimal for \(v\). Proposition~\ref{prop:single-trade} then gives
\(p(v,\mathbf z)=p=-F'(x)\). Hence the derivative of
\(x\mapsto v(x,F(x))\) is zero, and \(F\) is a level curve of \(v\). This
proves
\(\Lambda_{u_1}\boxplus\Lambda_{u_2}\subseteq\Lambda_v\).

For the reverse inclusion, fix \(\mathbf z\in\R_{++}^2\). The continuous
function
\(\mathbf r_1\mapsto u_1(\mathbf r_1)u_2(\mathbf z-\mathbf r_1)\)
vanishes on the boundary of the allocation rectangle and is positive in its
interior, so it has an interior maximizer. At that maximizer, the two
indifference prices agree. By Lemma~\ref{lem:infimal-convolution-regularity},
equality of the two slopes identifies this allocation as the unique
minimizing split in the relevant infimal convolution. Hence the
corresponding aggregate frontier passes through \(\mathbf z\), and the
preceding paragraph shows that it is a level curve of \(v\). Since distinct level curves of a strictly increasing utility
cannot intersect, every positive level curve of \(v\) belongs to the
aggregate family. Thus the reverse inclusion holds.

\section{Proofs for Section \ref{sec:tr}}\label{ec:pf6}

\subsection{Proof of Theorem \ref{thm:interior-feasible-reachable}}
We first establish the following auxiliary lemma.
\begin{lemma}
\label{lem:local-valid-reachability}
Let \(\varphi\) satisfy Assumptions~\eqref{as:1}--\eqref{as:3}.
Fix \(c>0\) and a total inventory \(\mathbf r_T\), and write the corresponding liquidity curve as $\mathbf z(x)=(x,\ell_{\varphi,c}(x))$ with $x>0$. We use coordinates \((x,\mathbf a)\) for market states on this curve, with
\[
\mathbf r_0=\mathbf z(x),\qquad
\mathbf r_1=\mathbf a,\qquad
\mathbf r_2=\mathbf r_T-\mathbf z(x)-\mathbf a.
\]
Thus, \(x\) is the CFMM's inventory of asset \(X\), \(\mathbf a\) is
agent~1's inventory, and agent~2's inventory is determined by total
inventory.
Define
\[
\Omega := \left\{
(x,\mathbf a):x>0,\;
\mathbf a\in\R_{++}^2,\;
\mathbf r_T-\mathbf z(x)-\mathbf a\in\R_{++}^2
\right\}.
\]
For every \((x,\mathbf a)\in\Omega\), there is a neighborhood \(V\) in
\(\Omega\) such that every point of \(V\) is reachable from \((x,\mathbf a)\) by finitely many valid trades whose intermediate states remain in \(\Omega\).
\end{lemma}

\begin{proof}
Fix \((x,\mathbf a)\in\Omega\), and write
\(\ell:=\ell_{\varphi,c}\).
Agent~2 can move the CFMM from
\(\mathbf z(x)\) to \(\mathbf z(x')\), for any sufficiently close \(x'\),
by submitting the valid trade $\mathbf h=\mathbf z(x')-\mathbf z(x)$. This leaves agent~1's holding unchanged.

It remains to generate arbitrary sufficiently small changes in agent~1's
holding while returning the CFMM to \(\mathbf z(x)\). For small
\(\delta\), consider the two-trade loop $x\to x+\delta\to x$, with agent~2 moving first and agent~1 second. Agent~1's net change is
\[
\mathbf z(x+\delta)-\mathbf z(x)
=
(\delta,\ell(x+\delta)-\ell(x)),
\]
which has a nonzero \(X\)-component.

To generate a vertical displacement, fix a small \(\eta\neq0\). For \(b\)
near zero, consider the four-trade loop
\[
x\longrightarrow x+\eta
\longrightarrow x+\eta+b
\longrightarrow x+b
\longrightarrow x,
\]
where agents \(2,1,2,1\) execute the successive trades. Agent~1's net
change is
\[
\mathbf z(x+\eta)-\mathbf z(x+\eta+b)
+\mathbf z(x+b)-\mathbf z(x).
\]
Its \(X\)-component is zero, and its \(Y\)-component is
\[
m(b)
=
\ell(x+\eta)-\ell(x+\eta+b)
+\ell(x+b)-\ell(x).
\]
We have \(m(0)=0\) and $m'(0)=-\ell'(x+\eta)+\ell'(x)\neq0$ for sufficiently small \(\eta\neq0\), because \(\ell'\) is strictly
increasing. The inverse function theorem therefore lets the four-trade
loop generate any sufficiently small vertical displacement.

Given a sufficiently small
\(\Delta=(\Delta_x,\Delta_y)\), first use the two-trade loop with
\(\delta=\Delta_x\). It produces $(\Delta_x,\ell(x+\Delta_x)-\ell(x))$.
Then use the four-trade loop to add the vertical amount
$\Delta_y-\bigl[\ell(x+\Delta_x)-\ell(x)\bigr]$.
After shrinking the neighborhood if necessary, every intermediate CFMM
inventory and agent holding remains strictly positive. Combining this
construction with the initial small move of the CFMM proves the claimed
local reachability.
\end{proof}

\begin{proof}[Proof of Theorem \ref{thm:interior-feasible-reachable}]
Let \(c=\varphi(\mathbf r_0^{\mathrm{in}})\). Write the liquidity curve as
\(\mathbf z(x)=(x,\ell_{\varphi,c}(x))\). By the coordinate identification
in Lemma~\ref{lem:local-valid-reachability}, the interior feasible states
with this invariant and total inventory are exactly the points
\((x,\mathbf a)\in\Omega\).

The set \(\Omega\) is path connected. For each admissible \(x\), the possible
holdings of agent~1 form the rectangle
\(0<\mathbf a<\mathbf r_{T}^{\mathrm{in}}-\mathbf z(x)\) coordinatewise, and the
admissible values of \(x\) form an interval. Any point can first be joined
within its rectangle to
\(\frac12(\mathbf r_{T}^{\mathrm{in}}-\mathbf z(x))\), and these midpoint holdings
can then be joined as \(x\) varies.

Declare two points of \(\Omega\) equivalent if one is reachable from the
other by a finite valid trading sequence whose intermediate states remain in
\(\Omega\). Valid trades are reversible, so this is an equivalence relation.
Lemma~\ref{lem:local-valid-reachability} makes every equivalence class open;
its complement is a union of other open classes and is therefore open.
Connectedness of \(\Omega\) implies that there is only one equivalence
class. The initial and target states are consequently connected by a finite
valid trading sequence with interior intermediate states.
\end{proof}

\subsection{Proof of Corollary \ref{cor:classical-eq}}
We first record the implications of the Walrasian equilibrium in the pure exchange economy. Market clearing gives $\mathbf r_1^{\mathrm W}+\mathbf r_2^{\mathrm W}=
\mathbf r_1^{\mathrm{in}}+\mathbf r_2^{\mathrm{in}}$.
Moreover, each Walrasian allocation \(\mathbf r_i^{\mathrm W}\) is interior.
Indeed, agent \(i\) can attain the strictly positive utility
\(u_i(\mathbf r_i^{\mathrm{in}})\) by keeping her interior initial endowment,
whereas every boundary allocation has utility zero.
Both components of \(\mathbf q^{\mathrm W}\) are strictly positive. If one
asset had zero price, an agent could increase her holding of that asset
without violating her budget constraint, and Assumption~\eqref{as:2} would
prevent her utility-maximization problem from having a solution.

We now return to the CFMM economy. Since the CFMM inventory is unchanged,
Walrasian market clearing implies that
\(\mathbf R^{\mathrm W}\) is feasible relative to
\(\mathbf R^{\mathrm{in}}\). By the assumed initial price alignment and the Walrasian first-order conditions,
\[
p(u_1,\mathbf r_1^{\mathrm W})
=
p(u_2,\mathbf r_2^{\mathrm W})
=
\frac{q_1^{\mathrm W}}{q_2^{\mathrm W}}
=
p(\varphi,\mathbf r_0^{\mathrm{in}}).
\]
Hence \(\mathbf R^{\mathrm W}\) is a unilateral equilibrium by
Theorem~\ref{thm:1}. Since \(\mathbf R^{\mathrm W}\) is also interior and
feasible, Theorem~\ref{thm:interior-feasible-reachable} implies that it is
validly reachable from \(\mathbf R^{\mathrm{in}}\).

\subsection{Proof of Theorem \ref{thm:max-reach-limit}}
We first establish the following auxiliary lemma.
\begin{lemma}
\label{lem:alternating-price-bracketing}
Let $\varphi,u_1,u_2$ satisfy Assumptions~\eqref{as:1}--\eqref{as:4}.
Consider an alternating utility-maximizing sequence from an interior state. Suppose that, after agent~1 has traded,
$p_0^1=p_1^1<p_2^1$.  
Define, for $n\geq1$,
\[
A_n:=p_0^{2n-1}=p_1^{2n-1},\qquad
B_n:=p_0^{2n}=p_2^{2n},
\]
and set $B_0:=p_2^1$.
Then
$A_n<A_{n+1}<B_n<B_{n-1}$ for every $n\geq1$. The agents' holdings and the
CFMM inventory converge to an interior market state, and
$A_n$ and $B_n$ converge to the same limit. The analogous conclusions hold
when $p_2^1<p_0^1=p_1^1$, with all inequalities reversed.
\end{lemma}
\begin{proof}
The entire path  stays in a compact subset of $(\R_{++}^2)^3$. Total
inventory is conserved, and the CFMM remains on the fixed positive invariant
level $\varphi(\mathbf r_0^t)=\varphi(\mathbf r_0^0)$. Hence every
coordinate is bounded above, while the CFMM cannot approach an axis. The
zero trade is available whenever an agent moves, so each agent's utility is
nondecreasing along the path and remains strictly positive. Since utility is
zero on the coordinate axes, neither agent can approach the boundary.

A single best response brackets prices. If the CFMM price is below the
active agent's indifference price, the agent buys $X$. The CFMM price rises,
the agent's price falls, and the two post-trade prices coincide. Their common
value therefore lies strictly between the pre-trade prices. The case in
which the CFMM price is above the agent's price is symmetric. This follows
from Proposition~\ref{prop:single-trade}~\eqref{trade:optimal-implies-price}
and~\eqref{trade:sign}.

Suppose that, after agent~1's $n$th trade,
$p_0=p_1=A_n<B_{n-1}=p_2$. Agent~2 buys $X$, so the one-trade bracketing
property gives $A_n<B_n<B_{n-1}$. Agent~1 then faces
$p_1=A_n<B_n=p_0$ and sells $X$, which gives
$A_n<A_{n+1}<B_n$. Thus
\[
A_n<A_{n+1}<B_n<B_{n-1},\qquad n\geq1.
\]
Consequently, $(A_n)$ is increasing, $(B_n)$ is decreasing, and both
converge.

After the first trade, agent~2 always buys $X$ and pays $Y$, whereas
agent~1 always sells $X$ and receives $Y$. Hence
$x_2^t$ and $y_1^t$ are nondecreasing, while $y_2^t$ and $x_1^t$ are
nonincreasing. These bounded coordinate sequences converge. Conservation
then makes the CFMM inventory converge as well. The compact-interior
argument above shows that the limiting state is interior.

The CFMM price converges by continuity. Since $(A_n)$ and $(B_n)$ are its
odd and even subsequences, they must have the same limit, say $L$. Between
successive trades, agent~1's price equals the relevant $A_n$, and agent~2's
price equals the relevant $B_n$. Thus all three prices converge to $L$.
The case with the reverse initial ordering follows by reversing all
inequalities and interchanging buying with selling.
\end{proof}

\begin{proof}[Proof of Theorem \ref{thm:max-reach-limit}]
Assume that agent~1 moves first. Proposition~\ref{prop:single-trade} makes
the CFMM price equal to agent~1's indifference price immediately after that
trade. If agent~2 has the same price, every later utility-maximizing trade is
zero and the result is immediate. Otherwise, one of the two strict orderings
$p_0^1=p_1^1<p_2^1$ or $p_2^1<p_0^1=p_1^1$ holds. By symmetry it is enough
to consider the first.

Lemma~\ref{lem:alternating-price-bracketing} shows that the two post-trade
price subsequences bracket one another, that the market state converges to
an interior limit, and that the price subsequences have the same limit. Continuity of the three price functions then yields the stated price equality at the limiting state.
\end{proof}

\subsection{Proof of Proposition \ref{prop:first-mover}}
Write $\mathbf r_i=(a,b)$, $\mathbf r_0=(x,\ell(x))$, and $\widetilde{\mathbf r}_0=(\widetilde x,\ell(\widetilde x))$, where $\ell=\ell_{\varphi,c}$ is the common liquidity curve. Define $P(s)=p(\varphi,s,\ell(s))=-\ell'(s)$. Proposition~\ref{prop:level-curve} implies that $P$ is strictly decreasing.

Let $\mathbf h^*$ and $\widetilde{\mathbf h}^*$ be the unique utility-maximizing trades from $(\mathbf r_0,\mathbf r_i)$ and $(\widetilde{\mathbf r}_0,\mathbf r_i)$, respectively. We compare their post-trade utilities.

First suppose $p_0>\widetilde p_0>p_i$. Since $P$ is strictly decreasing, $x<\widetilde x$. In both markets the agent sells $X$. Write $\widetilde{\mathbf h}^*=(m,\ell(\widetilde x+m)-\ell(\widetilde x))$ with $m>0$. The corresponding holding is $\widetilde{\mathbf r}_i^*=(a-m,b+\ell(\widetilde x)-\ell(\widetilde x+m))$.

From $\mathbf r_0$, consider selling the same amount $m$ of $X$. The valid trade is $\widehat{\mathbf h}=(m,\ell(x+m)-\ell(x))$, and the resulting holding is $\widehat{\mathbf r}_i=(a-m,b+\ell(x)-\ell(x+m))$. 
Since $x<\widetilde x$ and $-\ell'$ is strictly decreasing,
\[
\ell(x)-\ell(x+m)
=\int_x^{x+m}-\ell'(s)\,ds
>\int_{\widetilde x}^{\widetilde x+m}-\ell'(s)\,ds
=\ell(\widetilde x)-\ell(\widetilde x+m).
\]
Thus $\widehat{\mathbf r}_i$ has the same $X$-coordinate as $\widetilde{\mathbf r}_i^*$ but more $Y$. Strict monotonicity gives $u_i(\widehat{\mathbf r}_i)>u_i(\widetilde{\mathbf r}_i^*)$, and optimality of $\mathbf h^*$ gives $u_i(\mathbf r_i-\mathbf h^*)\geq u_i(\widehat{\mathbf r}_i)>u_i(\widetilde{\mathbf r}_i^*)$.

Now suppose $p_0<\widetilde p_0<p_i$. Then $x>\widetilde x$, and in both markets the agent buys $X$. Write $\widetilde{\mathbf h}^*=(-m,\ell(\widetilde x-m)-\ell(\widetilde x))$ with $m>0$. From $\mathbf r_0$, consider the valid trade $\widehat{\mathbf h}=(-m,\ell(x-m)-\ell(x))$. Since $x>\widetilde x$ and $-\ell'$ is strictly decreasing,
\[
\ell(x-m)-\ell(x)
=\int_{x-m}^{x}-\ell'(s)\,ds
<\int_{\widetilde x-m}^{\widetilde x}-\ell'(s)\,ds
=\ell(\widetilde x-m)-\ell(\widetilde x).
\]
The agent therefore pays less $Y$ for the same amount of $X$ from $\mathbf r_0$. The candidate holding from $\mathbf r_0$ has the same $X$-coordinate as $\widetilde{\mathbf r}_i^*$ but more $Y$. Strict monotonicity and optimality again give $u_i(\mathbf r_i-\mathbf h^*)>u_i(\widetilde{\mathbf r}_i^*)$.

\subsection{Proof of Proposition \ref{prop:first-mover-effect}}
Proposition~\ref{prop:single-trade}~\eqref{trade:sign} implies that if
\(p_0^0>p_i^0\), agent \(i\) sells \(X\) and the CFMM price falls, whereas
if \(p_0^0<p_i^0\), agent \(i\) buys \(X\) and the CFMM price rises.
Moreover, Proposition~\ref{prop:single-trade}~\eqref{trade:optimal-implies-price}
makes the post-trade common price lie strictly between the two pre-trade prices.

\proofpart{Part~(i)}
Suppose first that \(p_1^0,p_2^0<p_0^0\). Both agents want to sell \(X\).
If agent~2 moves first, then
$ p_2^0<p_0^{1,2}<p_0^0$.
By condition~\eqref{con:no-crossing},
\(p_0^{1,2}>p_1^0\), so $
p_0^0>p_0^{1,2}>p_1^0$.
Agent~1 has the same holding before her first trade on both sequences. When
she moves first she faces \(p_0^0\), whereas when she moves second she faces
\(p_0^{1,2}\). Proposition~\ref{prop:first-mover} therefore gives $
u_1(\mathbf r_1^{2,1})>u_1(\mathbf r_1^{2,2})$.
Interchanging the agents gives
$ u_2(\mathbf r_2^{2,2})>u_2(\mathbf r_2^{2,1})$.
The case \(p_1^0,p_2^0>p_0^0\) follows symmetrically.

\proofpart{Part~(ii)}
Suppose first that \(p_1^0<p_0^0<p_2^0\). Agent~1 wants to sell \(X\),
while agent~2 wants to buy it. If agent~2 moves first, she raises the CFMM
price to \(p_0^{1,2}\), with
$p_0^0<p_0^{1,2}<p_2^0$.
Since \(p_1^0<p_0^0<p_0^{1,2}\), agent~1 faces a more favorable selling
price when moving second. Proposition~\ref{prop:first-mover} therefore gives $
u_1(\mathbf r_1^{2,2})>u_1(\mathbf r_1^{2,1})$.

Similarly, if agent~1 moves first, she lowers the CFMM price to
\(p_0^{1,1}\), with $
p_1^0<p_0^{1,1}<p_0^0$.
Since \(p_0^{1,1}<p_0^0<p_2^0\), agent~2 faces a more favorable buying
price when moving second. Hence
$
u_2(\mathbf r_2^{2,1})>u_2(\mathbf r_2^{2,2})$.
The case \(p_2^0<p_0^0<p_1^0\) follows by interchanging the agents.

\subsection{Proof of Theorem \ref{thm:first-mover-advantage-limit}}

We first establish the following auxiliary lemma.

\begin{lemma}
\label{lem:second-mover-domination}
Let \(\varphi,u_1,u_2\) satisfy Assumptions~\eqref{as:1}--\eqref{as:4}.
Suppose that \(p_0^0<p_1^0\), \(p_0^0<p_2^0\), and that
condition~\eqref{con:no-crossing} holds. Let $
\ell:=\ell_{\varphi,\varphi(\mathbf r_0^0)}$ and $
m_i^{t,j}:=x_i^{t,j}-x_i^0$ for $t\in\mathbb N\cup\{\infty\}$.
Then, for
\((i,j)\in\{(1,2),(2,1)\}\), the holding
\[
\mathbf r_i^\gamma
:=
\bigl(
x_i^0+m_i^{\infty,j},\,
y_i^0-\ell(x_0^0-m_i^{\infty,j})+\ell(x_0^0)
\bigr)
\]
is attainable from \(\mathbf R^{0,j}\) through a single valid trade by
agent \(i\) and satisfies
$
x_i^\gamma=x_i^{\infty,j}$,  $
y_i^\gamma>y_i^{\infty,j}.
$
\end{lemma}

\begin{proof}
It suffices to prove the case \((i,j)=(1,2)\). Throughout the proof, write
\(\mathbf R^t:=\mathbf R^{t,2}\), including \(t=\infty\), and omit the
first-mover label from its components and from \(m_k^{t,2}\).
Quantities from the agent-1-first sequence retain their full labels.

\proofpart{Step 1: A one-agent price comparison}
For the fixed utility \(u_1\) and liquidity curve \(\ell\), consider two
interior one-agent optimal configurations labeled \(\alpha\) and \(\beta\).
Their combined CFMM--agent~1 holdings are
\(\mathbf r_{01}^\alpha\) and \(\mathbf r_{01}^\beta\), with
\[
\mathbf r_1^\alpha
=\mathbf r_{01}^\alpha-(x_0^\alpha,\ell(x_0^\alpha)),
\qquad
\mathbf r_1^\beta
=\mathbf r_{01}^\beta-(x_0^\beta,\ell(x_0^\beta)).
\]
Suppose that $
x_{01}^\beta\geq x_{01}^\alpha$ and 
$y_{01}^\beta\leq y_{01}^\alpha$,
with at least one inequality strict. We claim that
\begin{equation}
\label{eq:combined-inventory-price-order}
x_0^\beta>x_0^\alpha,
\qquad
p_0^\beta<p_0^\alpha.
\end{equation}
Indeed, if \(x_0^\beta\leq x_0^\alpha\), then
\[
x_1^\beta=x_{01}^\beta-x_0^\beta
\geq x_{01}^\alpha-x_0^\alpha=x_1^\alpha,
\]
\[
y_1^\beta=y_{01}^\beta-\ell(x_0^\beta)
\leq y_{01}^\alpha-\ell(x_0^\alpha)=y_1^\alpha,
\]
with at least one comparison strict. By Assumption~\eqref{as:4} and Proposition~\ref{prop:single-trade}, we have 
$p_0^\beta=p_1^\beta<p_1^\alpha=p_0^\alpha$,
contradicting the strict
decrease of the CFMM price along the liquidity curve.
Hence
\eqref{eq:combined-inventory-price-order} holds.

\proofpart{Step 2: Positive limiting net purchases}
In the present case, condition~\eqref{con:no-crossing} gives $
p_0^1=p_2^1=p_0^{1,2}<p_1^0$.
Lemma~\ref{lem:alternating-price-bracketing} implies that agent~1 subsequently always buys \(X\), whereas
agent~2 sells \(X\) at each subsequent trade. Since agent~1's first trade
is nonzero, we have $
m_1^\infty=x_1^\infty-x_1^0>0$.

We next show that \(m_2^\infty>0\). Suppose instead that
\(x_2^\infty\leq x_2^0\). Agent~2's initial nonzero best response strictly
improves her utility, and later trades do not decrease it. Hence $
u_2(\mathbf r_2^\infty)>u_2(\mathbf r_2^0)$.
Together with \(x_2^\infty\leq x_2^0\), it implies
\(y_2^\infty>y_2^0\). Assumption~\eqref{as:4} therefore gives
\begin{equation}
\label{eq:retained-purchase-price-contradiction}
p_0^\infty=p_2^\infty>p_2^0.
\end{equation}
The combined CFMM--agent~1 holdings satisfy
\[
\mathbf r_{01}^\infty
=\mathbf r_T^0-\mathbf r_2^\infty,
\qquad
\mathbf r_{01}^0
=\mathbf r_T^0-\mathbf r_2^0.
\]
Hence \(\mathbf r_{01}^\infty\) has weakly more \(X\) and strictly less
\(Y\) than \(\mathbf r_{01}^0\). At the limit, \(p_0^\infty=p_1^\infty\), so
Proposition~\ref{prop:single-trade}~\eqref{trade:price-implies-optimal}
implies that \((\mathbf r_0^\infty,\mathbf r_1^\infty)\) is the one-agent
optimum associated with \(\mathbf r_{01}^\infty\). Now consider the one-agent optimum obtained if agent~1 moves first from
the initial state.
Thus \((\mathbf r_0^{1,1},\mathbf r_1^{1,1})\) is the one-agent optimum
associated with \(\mathbf r_{01}^0\), with price \(p_0^{1,1}\). By \eqref{eq:combined-inventory-price-order}, we have $
p_0^\infty<p_0^{1,1}$.
Condition~\eqref{con:no-crossing} also gives $
p_0^{1,1}<p_2^0$.
Hence \(p_0^\infty<p_2^0\), contradicting
\eqref{eq:retained-purchase-price-contradiction}. Therefore
\(m_2^\infty>0\). Since agent~2's \(X\)-holding decreases after her first
trade,
\[
m_2^t\geq m_2^\infty>0,\qquad t\geq1,
\]
in particular whenever agent~1 trades.

\proofpart{Step 3: A cumulative payment comparison}
By conservation of \(X\), we have $x_0^t=x_0^0-m_1^t-m_2^t$.
If agent~1 buys an additional \(\delta>0\), her payment in \(Y\) is
\[
\ell(x_0^0-m_1^t-m_2^t-\delta)
-\ell(x_0^0-m_1^t-m_2^t)
=\int_{m_1^t}^{m_1^t+\delta}
p\!\left(\varphi,
x_0^0-v-m_2^t,
\ell(x_0^0-v-m_2^t)\right)\dd v.
\]
Since \(m_2^t\geq m_2^\infty\) and the CFMM price is strictly decreasing
along the liquidity curve, this payment is at least
\[
\int_{m_1^t}^{m_1^t+\delta}
p\!\left(\varphi,
x_0^0-v-m_2^\infty,
\ell(x_0^0-v-m_2^\infty)\right)\dd v.
\]
Summing over agent~1's trades and passing to the limit gives
\begin{align}
y_1^0-y_1^\infty
&\geq
\int_0^{m_1^\infty}
p\!\left(\varphi,
x_0^0-v-m_2^\infty,
\ell(x_0^0-v-m_2^\infty)\right)\dd v
\notag\\
&>
\int_0^{m_1^\infty}
p\!\left(\varphi,
x_0^0-v,
\ell(x_0^0-v)\right)\dd v
\notag\\
&=
\ell(x_0^0-m_1^\infty)-\ell(x_0^0),
\label{eq:second-mover-cumulative-cost}
\end{align}
where strictness follows from
\(m_1^\infty,m_2^\infty>0\) and the strict decrease of the CFMM price
along the liquidity curve.

Now consider the single trade from \(\mathbf R^{0,2}\) in which agent~1
buys \(m_1^\infty\) units of \(X\). The resulting holding is
\[
\mathbf r_1^\gamma
=
\bigl(
x_1^0+m_1^\infty,\,
y_1^0-\ell(x_0^0-m_1^\infty)+\ell(x_0^0)
\bigr).
\]
Thus \(x_1^\gamma=x_1^\infty\), while
\eqref{eq:second-mover-cumulative-cost} gives $
y_1^\gamma>y_1^\infty>0$.
Moreover, \(x_0^0-m_1^\infty>0\), so the trade is valid. This proves the
claim for \((i,j)=(1,2)\). The case \((i,j)=(2,1)\) follows by
interchanging the agents.
\end{proof}

\begin{proof}[Proof of Theorem \ref{thm:first-mover-advantage-limit}]
Without loss of generality, suppose that $
p_0^0<\min\{p_1^0,p_2^0\}$.
Apply Lemma~\ref{lem:second-mover-domination} with \((i,j)=(1,2)\).
The resulting holding \(\mathbf r_1^\gamma\) has the same \(X\)-coordinate
and strictly more \(Y\) than \(\mathbf r_1^{\infty,2}\). Hence, by strict
monotonicity and the optimality of agent~1's first trade, $
u_1(\mathbf r_1^{1,1})
\geq u_1(\mathbf r_1^\gamma)
>u_1(\mathbf r_1^{\infty,2})$.
Since utility of agent 1 is nondecreasing along \(\mathbf R^{t,1}\), we have $
u_1(\mathbf r_1^{\infty,1})
\geq u_1(\mathbf r_1^{1,1})
>u_1(\mathbf r_1^{\infty,2})$,
which proves~\eqref{eq:limiting-advantage-agent1}.
Interchanging the agents proves~\eqref{eq:limiting-advantage-agent2}.
\end{proof}

\end{document}